%% file: ms.tex
\documentclass[conference]{IEEEtran}
\usepackage{cite}
\input{defines}

\renewcommand{\baselinestretch}{0.9852}

\ifCLASSINFOpdf
\else
\fi
\begin{document}
%
\title{Finding events in temporal networks:\\ Segmentation meets densest-subgraph discovery}



%
\author{\IEEEauthorblockN{Polina Rozenshtein\IEEEauthorrefmark{1},
Francesco Bonchi\IEEEauthorrefmark{2},
Aristides Gionis\IEEEauthorrefmark{1},
Mauro Sozio\IEEEauthorrefmark{3} and
Nikolaj Tatti\IEEEauthorrefmark{4}}
\IEEEauthorblockA{\IEEEauthorrefmark{1}Aalto University, Espoo, Finland}
\IEEEauthorblockA{\IEEEauthorrefmark{2}ISI Foundation, Turin, Italy and Eurecat, Barcelona, Spain}
\IEEEauthorblockA{\IEEEauthorrefmark{3}Telecom ParisTech University, Paris, France}
\IEEEauthorblockA{\IEEEauthorrefmark{4}F-Secure Corp., Helsinki, Finland}}


\maketitle \sloppy

\begin{abstract}
\input{abstract}
\end{abstract}


%
\IEEEpeerreviewmaketitle

\input{intro}

\input{problem}

\input{approxDP}
\input{overlap}
\input{experiments}

\input{related}
\input{concl}
\section*{Acknowledgments}
Part of this work was done while the first author was visiting ISI Foundation. 
This work was partially supported by
three Academy of Finland projects  (286211, 313927, and 317085), 
and the EC H2020 RIA project ``SoBigData'' (654024). 
We thank the anonymous reviewers for their valuable comments.
%
%



%

\balance

\bibliographystyle{IEEEtran}
\bibliography{refs-brief}

\input{appendix}

\end{document}

%% file: defines.tex
\usepackage{booktabs} 
\usepackage{amsfonts}
\usepackage{amsthm}
\usepackage{amsmath}
\usepackage[ruled,linesnumbered]{algorithm2e}
\usepackage{mathtools}
\usepackage{xspace}
\usepackage{booktabs}
\usepackage{multirow}
\usepackage{tikz}
\usepackage{hyperref}
\usepackage{balance}
\usetikzlibrary{fit,external,positioning}
\usetikzlibrary{graphs}
\usetikzlibrary{external}

\usepackage[utf8]{inputenc}
\usetikzlibrary{babel}
\pgfdeclarelayer{bg}    
\pgfsetlayers{bg,main}

\newcommand{\set}[1]{\left\{#1\right\}}

\newcommand{\abs}[1]{{\left|#1\right|}}

\newcommand{\np}{\textbf{NP}}

\newcommand{\bigO}{\mathcal{O}}

\newcommand{\ourproblem}{\textsc{$k$-Densest-Episodes}}

\newtheorem{problem}{Problem}
\newtheorem{definition}{Definition}

\newtheorem{proposition}{Proposition}

\newtheorem{lemma}{Lemma}

\newcommand{\ind}{\mathbb{I}}

\newcommand{\ints}{\mathcal{I}}

\newcommand{\estat}{E'}
\newcommand{\etemp}{E}
\newcommand{\points}{m}
\newcommand{\tstamps}{\mathcal{T}}
\newcommand{\polylog}{\textrm{polylog}}
\newcommand{\ddens}{d^*}
\newcommand{\dens}{d}
\newcommand{\aprxddens}{d_a^*}

\newcommand{\gdens}{\dens_a}
\newcommand{\g}{\deg_a}
\newcommand{\mg}{\chi}

\newcommand{\dpgain}{\ensuremath{\mathrm{gain}}\xspace}
\newcommand{\graphs}{\mathcal{G}}
\newcommand{\stgraphs}{\mathcal{H}}
\newcommand{\dvrs}{\ensuremath{\mathrm{cover}}\xspace}
\newcommand{\find}{\texttt{Find}}
\newcommand{\approxdp}{\texttt{ApproxDP}\xspace}
\newcommand{\sprs}{\texttt{SPRS}\xspace}
\newcommand{\approxdens}{\texttt{ApprDens}\xspace}
\newcommand{\approxgendens}{\texttt{ApprGenDens}\xspace}
\newcommand{\profit}{\ensuremath{\mathrm{profit}}\xspace}

\newcommand{\edp}{\ensuremath{\epsilon_{\text{\sc dp}}}\xspace}
\newcommand{\eds}{\ensuremath{\epsilon_{\text{\sc ds}}}\xspace}

\SetKwInput{KwData}{Input}
\SetKwInput{KwResult}{Output}

\DeclarePairedDelimiter{\ceil}{\lceil}{\rceil}

\newcommand{\dataset}[1]{\textsl{#1}}
\newcommand{\facebook}{\dataset{Facebook}\xspace}
\newcommand{\twitter}{\dataset{Twitter}\xspace}
\newcommand{\twitterH}{\dataset{Twitter\#}\xspace}
\newcommand{\students}{\dataset{Students}\xspace}
\newcommand{\enron}{\dataset{Enron}\xspace}
\newcommand{\synone}{\dataset{Synthetic1}\xspace}
\newcommand{\syntwo}{\dataset{Synthetic2}\xspace}

\newcommand{\segment}{\textsc{kGapprox}\xspace}
\newcommand{\segmentCover}{\textsc{kGCvr}\xspace}
\newcommand{\segmentOpt}{\textsc{Optimal}\xspace}
\newcommand{\segmentDPopt}{\textsc{kGoptDP}\xspace}
\newcommand{\segmentDSopt}{\textsc{kGoptDS}\xspace}
\newcommand{\staticGreedy}{\textsc{StaticGreedy}\xspace}

\newcommand{\atleastk}{\emph{densest-at-least-$k$-subgraph}\xspace}

%% file: abstract.tex
In this paper we study the problem of
discovering a timeline of events in a temporal network.
We model events as dense subgraphs that occur within intervals of network activity.
We formulate the event-discovery task as an optimization problem,
where we search for a partition of the network timeline into $k$ non-overlapping intervals,
such that the intervals span subgraphs with maximum total density.
The output is a sequence of dense subgraphs along with corresponding time intervals,
capturing the most interesting events during the network lifetime.

A na\"ive solution to our optimization problem has polynomial but prohibitively high running time complexity. We adapt existing recent work on dynamic densest-subgraph discovery and approximate dynamic programming to design a fast approximation algorithm. Next, to ensure richer structure, we adjust the problem formulation to encourage coverage of a larger set of nodes. This problem is \np-hard even for static graphs. However, on static graphs a simple greedy algorithm leads to approximate solution due to submodularity. We extended this greedy approach for the case of temporal networks. However, the approximation guarantee does not hold. Nevertheless, according to the experiments, the algorithm finds good quality solutions. 

%% file: intro.tex
\section{Introduction}
Real-world networks are highly dynamic in nature,
with new relations (edges) being continuously established among entities (nodes), and old relations being broken.
Analyzing the temporal dimension of networks
can provide valuable insights about their structure and function,
for instance, it can reveal temporal patterns, concept drift, periodicity, temporal events, etc.
%
%
In this paper we focus on the problem of
\emph{finding dense subgraphs}, a fundamental graph-mining primitive.
Applications include community detection in social networks
\cite{chen2012dense, ditursi2017local, taylor2017super},
gene expression and drug-interaction analysis in bioinformatics
\cite{fratkin2006motifcut, saha2010dense}, graph compression and summarization~\cite{feder1995clique, karande2009speeding, hernandez2012compressed},
spam and security-threat detection \cite{gibson2005discovering, beutel2013copycatch},
and more.

When working with temporal networks one has first to define how to deal with the temporal dimension,
i.e., how to identify which are the temporal intervals in which the dense structures should be sought.
Instead of defining those intervals a-priori, in this paper
\emph{we study the problem of automatically identifying the intervals
that provide the most interesting structures}.
We consider a subgraph interesting if it boasts high density.
As a result, we are able to discover a sequence of dense subgraphs in the temporal network,
capturing the evolution of interesting events that occur during the network lifetime.
As a concrete example, consider the problem 
of \emph{story identification} in online social media~\cite{angel2012dense,DBLP:conf/icwsm/BalalauCS18}:
the main goal is to automatically discover emerging stories by finding
dense subgraphs induced by some entities,
such as twitter hash-tags, co-occurring in a social media stream.
In our case, we are additionally interested in understanding how the stories evolve over time.
For instance, as one story wanes and another one emerges,
one dense subgraph among entities dissipates and another one appears.
Thus, by segmenting the timeline of the temporal network into intervals,
and identifying dense subgraphs in each interval,
we can capture the evolution and progression of the main stories over time.

As another example, consider a collaboration network,
where a sequence of dense subgraphs in the network can reveal information
about the main trends and topics over time,
along with the corresponding time intervals.

\smallskip
\noindent \textbf{Challenges and contributions.}
The problem of finding the $k$ densest subgraphs in a static graph
has been considered in the literature from different perspectives.
One natural idea is to iteratively (and greedily)
find and remove the densest subgraphs~\cite{tsourakakis2013denser}.
More recent works consider finding $k$ densest graphs with limited
overlap~\cite{balalau2015finding, galbrun2016top}.
However, these approaches do not generalize to temporal networks.

For temporal networks, to our knowledge,
there are only few papers that consider the task of finding temporally-coherent densest subgraphs.
The most similar to our work
aims at  finding a heavy subgraph present in all, or~$k$, snapshots~\cite{semertzidis2016best}.
Another related work focuses on finding a dense subgraph
covered by $k$ scattered intervals in a temporal network~\cite{rozenshtein2017finding}.
Both of these methods, however, find a single densest subgraph.


In this paper, instead, we aim at producing a segmentation of the temporal network that
($i$) captures dense structures in the network;
($ii$) exhibits temporal cohesion;
($iii$) spans the whole history of the network; and
($iv$) is amenable to direct inspection and temporal interpretation.
Towards this goal we formulate the problem of
\ourproblem,
which requires to find a partition of the temporal domain into $k$ non-overlapping intervals,
such that the intervals span subgraphs with maximum total density.
The output is a sequence of dense subgraphs along with corresponding time intervals,
capturing the most interesting events during the network lifetime.

A na\"ive solution to this problem has polynomial but prohibitively-high running-time complexity.
Thus, we adapt existing recent work on dynamic-densest subgraph~\cite{epasto2015efficient}
and approximate dynamic programming~\cite{tatti15segmentation} to design a fast approximation algorithm.

Next we shift our attention to encouraging coverage of a larger set of nodes,
so as to produce richer, more interesting structures.
The resulting new problem formulation turns out to be \np-hard even for the case of static graphs.
However, on static graphs a simple greedy algorithm leads to approximate solution thanks  to the submodularity of the objective function.
Following this observation, we extended this greedy approach for the case of temporal networks.
Despite the fact that the approximation guarantee does not carry on when generalizing to the temporal case,
our experimental evaluation indicates that the method produces solutions of very high quality.

The contributions of this paper are summarised as follows:
\begin{itemize}
\item We introduce (Section~\ref{section:problem}) the \ourproblem\ problem and show that it has a polynomial time exact algorithm, which is however cubic thus unpractical.

\item  By leveraging recent work on dynamic densest subgraph
and approximate dynamic programming we achieve a fast  algorithm with approximation guarantees (Section~\ref{section:dp}).

\item We then (Section~\ref{section:overlap}) extend the problem formulation to
encourage coverage of a larger set of nodes. We show that the resulting problem is \np-hard even for the case of static graph. However, we show on static graphs a simple greedy algorithm leads to approximate solution due to submodularity; then we extend this greedy approach for the case of temporal networks.

\item Experiments on synthetic and real-world datasets (Section~\ref{section:experiments}), and a case study on Twitter data (Section \ref{sec:casestudy}) confirm that our methods are efficient and produce meaningful and high-quality results.
\end{itemize}

%% file: problem.tex
\section{Problem formulation}
\label{section:problem}
We are given a \emph{temporal graph} $G = (V,\tstamps,\tau)$, where $V$ denotes the set of nodes, $\tstamps = [0, 1, \ldots, t_{max}] \sqsubseteq \mathbb{N}$ is a discrete time domain, and $\tau: V  \times V \times  \tstamps\rightarrow \{0,1\}$ is a function defining for each pair of nodes $u,v \in V$ and each timestamp $t \in \tstamps$ whether edge $(u,v)$ exists in $t$. We denote $E = \{(u,v,t) \mid \tau(u,v,t) = 1 \}$ the set of all temporal edges. Given a temporal interval $T=[t_1, t_2] \sqsubseteq \tstamps$, let $G[T] = (V[T],E[T])$ be the subgraph induced by the set of temporal edges $E[T] = \set{(u,v) \mid (u,v,t) \in E \wedge  t \in T}$.

\begin{definition}[Episode]
Given a temporal graph $G = (V,\tstamps,\tau)$ we define an episode as a pair $(I,H)$ where $I \sqsubseteq \tstamps$ is a temporal interval and $H$ is a subgraph of $G[I]$.
\end{definition}
Our goal is to find a set of interesting episodes along the lifetime of the temporal graph. In particular, our measure of interestingness is the density of the subgraph in the episodes.
We adopt the widely-used notion of density of a subgraph $H = (V(H),E(H))$ as the average degree of the nodes in the subgraph, i.e.,  $d(H)=\frac{|E(H)|}{|V(H)|}$.  Observe that
this definition is not the only choice, however, such a
notion of density enjoys the following nice properties: It can be optimized exactly~\cite{goldberg1984finding} and approximated efficiently~\cite{charikar2000greedy}, while a densest subgraph can be computed in real-world graphs containing up to tens of billions of edges~\cite{DBLP:conf/www/DanischCS17}.




\begin{problem}[\ourproblem \label{p1}]
Given a temporal graph $G = (V,\tstamps,\tau)$ and an integer $k \in \mathbb{N}$, find a set of $k$ episodes
$S=\set{(I_\ell, H_\ell)}$, for $\ell = 1, \ldots, k$ such that the $\{I_\ell\}$ are disjoint intervals and
$\sum_{\ell = 1}^k \dens(H_\ell)$ is
	maximized.
\end{problem}

A solution for Problem~\ref{p1} can be computed in polynomial time.  To
see this, let $S^*$ be an optimum solution and let $\ints(S^*)=\{I_\ell\, , \ell = 1, \ldots, k\}$ and $\graphs(S^*)=\{H_\ell\, , \ell = 1, \ldots, k\}$.
We can assume without loss of generality that the
union of the  intervals in $\ints(S^*)$ results in the set of time stamps $\tstamps$, that
is, $\ints(S)$ is a $k$-segmentation of $\tstamps$. Moreover, a graph $H_\ell
\in \graphs(S^*)$ is the densest subgraph of $G(I_\ell)$, and can be found in
$\bigO(nm \log n)$ time~\cite{goldberg1984finding,orlin2013max} or in $\bigO(nm\log(n^2/m))$ time~\cite{gallo1989fast} (where $n$ and $m$ denote the number of nodes and edges in $G(I_\ell)$ respectively). The optimal segmentation can be solved
with a standard dynamic programming approach, requiring $\bigO(km^2)$ steps~\cite{bellman2013dynamic}. This brings the total running time to $\bigO(km^3n \log n)$ or $\bigO(km^3n\log(n^2/m))$.

%% file: approxDP.tex
\section{Approximate dynamic programming}
\label{section:dp}

The simple algorithm discussed in the previous section has a running time, which is prohibitively expensive for large graphs. In this section we develop a fast algorithm with approximation guarantees.

The derivations below closely follows the ones in \cite{tatti15segmentation}, which improves \cite{guha2001data}. However, we cannot use those results directly: both papers work with minimization problems, while leveraging the fact that the profit of an interval is not less than the profit of its subintervals (monotone non-decreasing). 
In contrast, our problem can be viewed as a minimization problem with monotone non-increasing profit function.

Given a time interval $T = [t_1, t_2]$, 
let us write $\ddens(T)=\linebreak\max_{H\subseteq G(T)} d(H)$. 
For simplicity, we define $\ddens([t_1, t_2]) = 0$ if $t_2<t_1$.
Problem~\ref{p1} is now a classic $k$-segmentation problem of $\tstamps$ maximizing
the total sum of scores $\ddens(T)$ for individual time intervals.
For notation simplicity, we assume that the all timestamps $\tstamps$ are enumerated by integers from $1$ to $r$.

Let $o[i, \ell]$ be the profit of optimal $\ell$-segmentation using only the first $i$
time stamps. It holds:
\[
	o[i, \ell] = \max_{j < i} o[j, \ell - 1] + \ddens(j + 1, i),
\]
and $o[i, k]$ can be computed recursively. Denote the approximate profit of optimal $\ell$-segmentation as $s[i, \ell]$. The main idea behind the
speed-up is not to test all possible values of $j$. Instead, we are going to
keep a small set of candidates, denoted by $A$, and only use those values for
testing. The challenge is how to keep $A$ small enough while at the same time
guarantee the approximation ratio. The pseudo-code achieving this balance
is given in Algorithm~\ref{alg:approx}, while a subroutine that keeps the candidate list short is given in Algorithm~\ref{alg:spars}. Algorithm~\ref{alg:approx} executes a standard dynamic programming search: it assumes that partition of $i'<i$ first data points into $\ell-1$ intervals is already calculated and finds the best last interval $[a,i]$ for partitioning of $i$ first points into $l$ intervals. However, it considers not all possible candidates $[a,i]$, but only a sparsified list, which guarantees to preserve a quality guarantee. The sparsified list is built for a fixed number of intervals $\ell$ starting from empty list. Intuitively, it keeps only candidates $A=[a_j]$ with significant difference $s[a_j,\ell-1]$. Significance of the difference depends on the current best profit $s[i,\ell]$: the larger the value of the solution found, the less cautions we can be about lost candidates and the coarser becomes $A$. Thus, we need to refine $A$ by Algorithm~\ref{alg:spars} after each processed $i$.

\begin{algorithm}[t]
	\KwData{number of intervals $k$, parameter $\epsilon$}
	\KwResult{approximate solution $s[i,\ell]$ for $i \in [1, r]$, $\ell \in [1, k]$}
	
	\lFor {$i=1,\dots, r$} {$s[i,1]=\ddens([1,i])$}	
	\For {$\ell=2,\dots, k$}
	{
		$A=[]$\;
		\For {$i=1,\dots, r$}
		{
			add $i$ to $A$\;
			$s[i,\ell]=\max \{s[i-1,\ell],s[i,\ell-1], \max_{a\in A}  (s[a-1,\ell-1]+\ddens([a,i]))\}$\;
			
			$A=\sprs(A, s[i,\ell], \ell, \epsilon)$
		}
	}
	\Return $s$
	\caption{$\approxdp(k,\epsilon)$, computes $k$-segmentation with $\epsilon$-approximation guarantee}
	\label{alg:approx}
\end{algorithm}
\begin{algorithm}[t]
	\KwData{candidates $A$, sparsification factor $\sigma=s[i,\ell]$, current number of intervals $\ell$, approximation parameter $\epsilon$}
	\KwResult{sparsified $A$}
	
	$\delta = \sigma \frac{\epsilon}{k+\ell\epsilon}$\;
	$j=1$\;
	\While {$j < |A|-1$}
	{
		\lIf{$s[a_{j+2}, \ell-1] - s[a_j,\ell-1]\leq \delta$}
		{remove $a_{j+1}$ from $A$}
		\lElse{$j = j+1$}
	}
	\Return $A$
	\caption{$\sprs(A, \sigma, \ell, \epsilon)$, a subroutine keeping the candidate list short.}
	\label{alg:spars}

\end{algorithm}

Let us first prove that \approxdp yields an approximation guarantee, assuming that
$\ddens(\cdot)$ is calculated exactly.

\begin{proposition}\label{mainthrem}
Let $s[i, \ell]$ be the profit table constructed by \approxdp$(k, \epsilon)$. Then
	$s[i,\ell](\frac{\ell\epsilon}{k}+1)\geq o(i,\ell)$.
\end{proposition}

To prove the final result, let us first fix $\ell$ and
let $A_i$ be the set of candidates in the beginning of round $i$.
Let $\delta_i$ be the value of $\delta$ in Algorithm~\ref{alg:spars}, called on iteration $i$.

\begin{lemma}\label{baselemma}
	For every $b\in [1,i-1]$, there is $a_j, a_{j+1}\in A_i$ with $a_{j}\leq b \leq a_{j+1}$, such that
	\[s[a_j-1,\ell-1]+\ddens([a_j,i])\geq s[b-1,\ell-1]+\ddens([b,i])-\delta_{i-1}.\]
\end{lemma}
\begin{proof}

We say that a list of numbers $A = \set{a_j}$ is $i$-dense, if
\begin{equation}
\label{eq:dense}
	s[a_{j+1}-1,\ell-1] - s[a_j-1,\ell-1]\leq \delta_{i-1} \text{ or } a_{j+1} = a_j+1,
\end{equation}
for every $a_j \in A$ with $j < \abs{A}$.
We first prove by induction over $i$ that $A_i$ is $i$-dense.
	
Assume that $A_{i - 1}$ is $(i - 1)$-dense. \sprs never deletes the last element, so $i - 1 \in A_{i - 1}$,
and $A_{i - 1} \cup \set{i}$ is $(i - 1)$-dense.
Note that $\delta_{i-2}\leq \delta_{i-1}$, because $s[i,\ell]$ is monotonic, $s[i,\ell]\geq
s[i-1,\ell]$, due to explicit check on line 5 of \approxdp. Thus, $A_{i - 1} \cup \set{i}$ is $i$-dense.
Since $A_i=\sprs(A_{i-1}\cup\{i\})$, and $\sprs$ does not create
gaps larger than $\delta_{i-1}$, $A_i$ is $i$-dense.

Let $a_j$ be the largest element in $A_i$, such that $a_j\leq b$. Then either
$a_j\leq b<a_{j+1}$ or $b=a_{|A_i|}$ and $j=a_{|A_i|}$. Due to monotonicity,
$s[a_{j+1}, \ell-1]\geq s[b,\ell-1]$ and gives $s[b-1, \ell-1]-s[a_j-1, \ell-1]\leq
\delta_{i-1}$ for the first case. The second case is trivial.
	
Due to monotonicity $\ddens([a_j,i])\geq \ddens([b,i])$. This concludes the
proof.  \end{proof}

We can now complete the proof.

\begin{proof}[Proof of Proposition~\ref{mainthrem}]
We will prove the result with induction over $\ell$.
Let $\alpha=(1+\frac{\epsilon}{k}(\ell-1))$. Let $b$ be the starting point of the last interval of optimal solution $o[i,\ell]$, and let $a_j$
as given by Lemma~\ref{baselemma}.
Upper bound $\delta_{i-1}=s[i-1, \ell]\frac{\epsilon}{k+\epsilon(\ell-1)}\leq s[i, \ell]\frac{\epsilon}{k+\epsilon(\ell-1)} = s[i, \ell]\frac{\epsilon}{\alpha k}$.
Then
\[
	\begin{split}
	\alpha s[i,\ell]&\geq \alpha (s[a_j-1,\ell-1] + \ddens([a_j,i]))\\
	& \geq \alpha (s[b-1,\ell-1] + \ddens([b,i]) - \delta_{i-1})\\
	& \geq o[b - 1,\ell-1] + \ddens([b,i]) - \alpha\delta_{i-1}\\
	& \geq o[i,\ell] - s[i, \ell]\frac{\epsilon}{k}.
	\end{split}
\]
As a result, $s[i,\ell](1+\frac{\epsilon}{k}\ell)\geq o[i,\ell])$.
\end{proof}

Let us now address the computational complexity.

\begin{proposition}
The running time of \approxdp is $\bigO(\frac{k^2}{\epsilon}r)$.
\label{prop:time}
\end{proposition}
\begin{proof}
Fix $i$ and $\ell$, and let $c_j = s[a_j, \ell]$, where $a_j \in A_i$.
Then $c_j$ is monotonically increasing sequence upper bounded by $s[i - 1, \ell]$, and having consecutive elements being at least $\delta_{i - 1}$ apart.
Counting conservatively, this leads to
\[
	|A_i| \leq 2+ \ceil[\Big]{\frac{s[i-1,\ell]}{\delta_{i-1}}}\leq 2 + \ceil[\Big]{\frac{k(1+\epsilon)}{\epsilon}} \in \bigO(k/\epsilon).
\]
Since we have $kr$ cells in $s$, the result follows.
\end{proof}

Since computing $\ddens$ requires $\bigO(nm \log n)$ time, this gives us a total running time
of $\bigO(nmr \frac{k^2}{\epsilon})$.
We further speed up our algorithm by approximating the value $\ddens$ by means of one of the approaches developed in~\cite{epasto2015efficient}. In particular, we employ the algorithm that maintains a $2(1 + \epsilon)$-approximate solution for the incremental densest subgraph problem (i.e. edge insertions only), while boasting a poly-logarithmic amortized cost. We shall refer to such an algorithm as \approxdens.

\approxdens allows us to efficiently maintain the approximate density of the densest subgraph $\ddens([a,i])$ for each $a$ in $A_i$ in \approxdp, as larger values of $i$ are processed and edges are added. Whenever we remove an item $a$ from $A_i$ in \sprs we also drop the corresponding instance of \approxdens.

From the fact that an approximate densest subgraph can be maintained with poly-logarithmic amortized cost, it follows that our algorithm boasts quasi-linear running time.

\begin{proposition}
\approxdp combined with \approxdens
runs in $\bigO(\frac{k^2}{\epsilon_1 \epsilon_2^2} m \log^2 n )$,
where $\epsilon_1$ and $\epsilon_2$ are the respective approximation parameters
for \approxdp and \approxdens.
\end{proposition}

\begin{proof}
Let $m_j$ be the number of edges added to the graph corresponding to $a_j$
before it is deleted.  The same argument as in the proof of
Proposition~\ref{prop:time} states that $\sum_i m_i \in \bigO(\frac{k^2}{\epsilon_1}m)$.
Theorem 4 in~\cite{epasto2015efficient} states that maintaining the graph
with $m_i$ edges requires $\bigO(m_i \epsilon_{2}^{-2} \log^2 n)$ time.
Combining these two results proves the proposition.
\end{proof}



When combining \approxdp with \approxdens, we wish to maintain the same approximation guarantee of \approxdens. Recall that \approxdp leverages the fact that the profit function is monotone non-increasing. Unfortunately, \approxdens does not necessarily yield a monotone score function, as the density of the computed subgraph might decrease when a new edge is inserted. This can be easily circumvented by keeping track of the best solution, i.e. the subgraph with highest density. The following proposition holds. 

%

\begin{proposition}
\approxdp combined with \approxdens yields a $2(1+\epsilon_1)(1 + \epsilon_2)$-approximation guarantee.
\end{proposition}

\begin{proof}
Let $\aprxddens(T)$ be the density of the graph returned by \approxdens for a
time interval $T$.
Let $O$ be the optimal $k$-segmentation, and let $q_1 = \sum_{I \in \ints(O)} \ddens(O)$
be its score. Let also $q_2 = \sum_{I \in \ints(O)} \aprxddens(O)$.
Let $q_3$ be the score of the optimal $k$-segmentation using $\aprxddens$, and let $q_4$
be the score of the segmentation produced by \approxdp. Then,
\[
	q_1 \leq 2(1 + \epsilon_2)q_2 \leq 2(1 + \epsilon_2)q_3 \leq 2(1 + \epsilon_2)(1 + \epsilon_1)q_4,
\]
completing the proof.
\end{proof}
We will refer to this combination of \approxdp with \approxdens as Algorithm \segment.

%% file: overlap.tex
\section{Encourage coverage}
\label{section:overlap}
Problem~\ref{p1} is focused on total density maximization, thus its solution can contain graphs which are dense, but union of their node sets cover only a small part of the network. Such segmentation is useful when we are interested in the densest temporally coherent subgraphs which can be understood as tight cores of temporal clusters. However, segmentations with larger but less dense subgraphs, covering a larger fraction of nodes, can be useful to get a high-level explanation of the whole temporal network. To allow for such segmentations we extend Problem~\ref{p1} to take node coverage into account.

Let $x_v(\graphs)=|\{G_i\in \graphs: v\in G_i,  G_i\in\graphs\}|$ be the number of subgraphs in $\graphs$, which include node $v$.

Here we consider a generalized cover functions of the shape
\[
	\dvrs(\graphs) = \sum_{v \in V}w(x_v(\graphs)),
\]
where $w$ is a non-negative non-decreasing concave function of $x_v(\graphs)$.
When $w(x_v(\graphs))$ is a 0-1 indicator function, function $\dvrs(\graphs)$
is a standard cover, which is intuitive and easy to optimize by greedy
algorithm. Another instance of the generalized cover function, inspired by
text-summarization research \cite{lin2011class}, is
$w(x_v(\graphs))=\sqrt{x_v(\graphs)}$. It ensures that the marginal gain of a
node decreases proportionally to the number of times the node is covered.


\begin{problem}\label{p2}
	Given a temporal graph $G=(V,E)$, integer $k$, parameter $\lambda\geq 0$. Find a $k$-segmentation
	$S=\{(I_i, G_i)\}$ of $G$, such that $\profit(S) = \sum_{G_i\in
	\graphs}\dens(G_i) + \lambda\times\dvrs(\graphs)$ is maximized.
\end{problem}


\begin{proposition}\label{prop:NPP2}
	There is no polynomial solution for Problem~\ref{p2} unless P=NP.
\end{proposition}
\begin{proposition}
	Function $\dvrs(\graphs)$ is a non-negative non-decreasing submodular function of subgraphs.
\end{proposition}
\begin{proof}
	For a fixed $v\in V$ function $x_v(\graphs)$ is non-decreasing modular (and submodular): for any set of subgraphs $X$ and a new subgraph $x$ holds that $x_v(X\cup \{x\})-x_v(X)=1$ if $v$ belongs to $x$ and does not belong to any subgraph in $X$. Otherwise $0$. By property of submodular functions, composition of concave non-decreasing and submodular non-decreasing is non-decreasing submodular. Function $\dvrs(\graphs)$ is submodular non-decreasing as a non-negative linear combination. Non-negativity follows from non-negativity of $w$.	
\end{proof}	

\subsection{K static densest subgraphs and generalized average degree}


Before going into the temporal segmentation, we briefly consider the static case:
\begin{problem}\label{stat}
	Given a static graph $H=(V,\estat)$, integer $k$, $\lambda\geq 0$. Find a set of $k$ subgraphs $\stgraphs = \{H_i \in H\}$, such that $\profit_{ST} = \sum_{H_i\in \stgraphs}\dens(H_i) + \lambda\cdot\dvrs(\stgraphs)$ is maximized.
\end{problem}

To solve this problem we can search greedily over subgraphs. Let $\stgraphs_{i-1} = \{H_1,\dots, H_{i-1}\}$, and define marginal node gain, given weight function $w$, as $\delta_v(H_i, \stgraphs_{i-1}\mid w) = w(x_v(\stgraphs_{i-1}\cup \{H_i\}))-w(x_v(\stgraphs_{i-1}))$. Then denote marginal gain of subgraph $H_i$ given already selected graphs $\stgraphs_{i-1}$ as
\begin{equation*}
\mg(H_i, \stgraphs_{i-1}\mid w)=\dens(H_i)+\lambda\sum_{v\in H_i}\delta_v(H_i, \stgraphs_{i-1}\mid w). \\
\end{equation*}

Greedy algorithm for Problem~\ref{stat} consequently builds the set $\stgraphs$ by adding $H_i$, which maximizes gain $\mg(H_i, \stgraphs_{i-1})$. If we can find $H_i$ optimally, such greedy gives $1-1/e$ approximation due to submodular maximization over cardinality constrains (see~\cite{nemhauser1978analysis} for this classic result, Euler's number $e\approx 2.71828$).

\begin{problem}\label{statgreedystep}
	Given a static graph $H=(V,\estat)$, a set of subgraphs $\stgraphs_{i-1} = \{H_1,\dots, H_{i-1}\}$, find graph $H_i \in H$, such that $\mg(H_i, \stgraphs_{i-1})$ is maximized.
\end{problem}

Before we proceed, we define a more general and simple version of Problem~\ref{statgreedystep}. First, we note that preselected subgraphs $\stgraphs_{i-1}$ contribute only to $\delta_v(H_i, \stgraphs_{i-1}\mid w)$ and this term does not change through iterations. Thus, once term $\delta_v(H_i, \stgraphs_{i-1}\mid w)$ is recalculated we can exclude $\stgraphs_{i-1}$ from consideration.

Next, we define a generalized degree as a function of nodes defined as $\g(v\mid H) = \sum_{u\in V\setminus\{v\}} a(v,u\mid H)$ with $a(v,u\mid H)\geq 0$.
Let $\ind(v,u\mid H)$ be $1$, if there is an edge between $u$ and $v$ in graph $H$ and $0$ otherwise.
If $a(v,u\mid H) = \ind(v,u\mid H)$, then $\g(v\mid H)=\deg(v\mid H)$---a degree of node $v$ in $H$.
Denote a half of the average generalized degree of graph $H$ as $\gdens(H)=\frac{1}{2|V(H)|}\sum_{v\in V(H)} \g(v\mid H)$.

\begin{problem}\label{maxgendegree}
	Given a static graph $H=(V,\estat)$ find graph $H_i=(V(H_i),\estat(H_i)) \subseteq H$, such that a half of the average generalized degree $\gdens(H_i)$ is maximized.
\end{problem}

If $a(v,u\mid H) = \ind(v,u\mid H)$, then the profit of Problem~\ref{maxgendegree} is the half of average degree $\gdens(H_i)=\dens(H_i)$.
On the other hand, when $a(v,u\mid H) = \ind(v,u\mid H) + 2\lambda |V(H)| \delta_v(H_i, \stgraphs_{i-1}\mid w)$, then Problem~\ref{maxgendegree} is equivalent to Problem~\ref{statgreedystep}. Note that in the latter case $a$ depends on the number of nodes in graph $H$.

We will continue analysis with Problem~\ref{maxgendegree}.


\begin{proposition}\label{prop:NPP5}
	There is no polynomial solution for Problem~\ref{maxgendegree} unless P=NP.
\end{proposition}

To solve Problem~\ref{maxgendegree} efficiently we can modify Charikar's algorithm for densest subgraphs~\cite{charikar2000greedy} and obtain $1/2$ approximation guarantee.

\begin{algorithm}[t]
	\KwData{static graph $H=(V,\estat)$}
	\KwResult{subgraph $\bar H\subseteq H $ which maximizes $\gdens(\bar H)$}
	$\bar H=H$\;
	\While{$H\not = \emptyset$}
	{
		$v = \arg\min_{v\in H} \g(v\mid H)$\;
		$H=H\setminus\{v\}$\;
		\lIf {$\gdens(H) > \gdens(\bar H)$} {$\bar H=H$}
	}	
	\Return $\bar H$
	\caption{\texttt {StaticGreedy}}
	\label{alg:statgreedy}
\end{algorithm}

\begin{proposition}\label{prop:statgreedy}
	Algorithm~\ref{alg:statgreedy} gives $1/2$ approximation for Problem~\ref{maxgendegree}.	
\end{proposition}

Time complexity of Algorithm~\ref{alg:statgreedy} is quadratic of number of nodes (not linear of the edges, like in the case of densest subgraph), as $|V(H)|$ decreases on each step and we need to update generalized degree of all nodes, not only neighbors of the removed node. With a Fibonacci heap time complexity is $\Theta(|V| + |V|(\log{|V|} + |V|))=\Theta(|V|^2)$.

\subsection{Incremental case}

Here we consider the setting of incremental updates for Problem~\ref{maxgendegree}, which may be not interesting by itself, but we will use it as a subroutine for temporal case.

Given a stream of incremental edge updates to graph $H$ we would like to find and keep up-to-date a subgraph $H_i$, which maximizes $\gdens(H_i)$ for some generalized degree function $\g(u,v\mid H_i)$.

To keep $H_i$ updated we can use the data structure and update procedure designed for the densest subgraph by Epasto et al.~\cite{epasto2015efficient}. In the full version of this paper we describe the approach of Epasto et al. and necessary modifications to handle generalized degree. We will refer to this extension as \approxgendens. The algorithm provides $2(1 + \epsilon)$-approximate generalized density densest subgraph using edge insertions.

Similar to the original algorithm, the generalization requires $O(|V|+|E|)$ of space, while running time increases: $O(\frac{|V|^2}{|E|}\epsilon^{-2}\log^2{D})$ amortized cost per edge insertion, with $D=O(|V|)$ is the maximum of average generalized degree.

\subsection{Greedy dynamic programming}

Similarly to Problem~\ref{p1}, we will use dynamic programming for Problem~\ref{p2}. However, as the problem is hard we have to rely on greedy choices of the subgraphs. Thus, the obtained solution does not have any quality guarantee.

Let $M[\ell,i]$ be the profit of $i$ first
points into $\ell$ intervals, let $C[\ell,i]$ be the set of subgraphs
$\graphs_{\ell}=\{G_1, \dots, G_{\ell}\}$ selected on these $\ell$ intervals, $1\leq
\ell \leq k$ and $0\leq i\leq \points$.

Define marginal gain interval $[j,i]$, given that $j-1$ are already segmented into $\ell - 1$ interval,
\begin{equation*}
	\dpgain([j,i], C[\ell-1, j-1]) = \max_{G'\subseteq G([j,i])}\mg(G', C[\ell-1, j-1]).
\end{equation*}

Dynamic programming recurrence:
\[
\begin{split}
	M[\ell,i]  = & \max_{1\leq j\leq i+1} M[\ell-1,j-1] \\
	           & + \dpgain([j,i], C[\ell-1, j-1]) \textrm{ for } 1 < \ell \leq k, \\
	M[1,i]   = & \ddens([0, i]) \textrm{ for } 0\leq i\leq \points, \\
	M[k',0]  = & 0 \textrm{ for } 1\leq k'\leq k.
\end{split}
\]

After filling this table, $M[k,m]$ contains the profit of k-segmentation with subgraph overlaps. $C[k,m]$ will contain selected subgraphs, the intervals and subgraphs can be reconstructed, if we keep track of the starting points of selected last intervals. Note, that profit $M[k,m]$ is not optimal, because the choice of subgraph $G_i$ depends on the interval and the previous choices, and there is a fixed order, in which we explore intervals.

We perform dynamic programming by approximation algorithm \approxdp, and the densest subgraph for each candidate interval is retrieved by \approxgendens. We refer to the resulting algorithm as \segmentCover.


To keep track on number of $x_v$ when we construct $\graphs$ we need to keep frequencies of each node. To avoid extensive memory costs, in the experiments we use Min-Count sketches.


%% file: experiments.tex
\section{Experiments}
\label{section:experiments}

We evaluate the performance of the proposed algorithms on synthetic graphs and real-world social networks.
The datasets are described below.
Unless specified, we post-process the output of all algorithms and report the optimal densest subgraphs in the output intervals. Our datasets and implementations are publicly available.\footnote{\url{https://github.com/polinapolina/segmentation-meets-densest-subgraph}}

\subsection{Synthetic data.}
We generate a temporal network with $k$ planted communities and a background network.
All graphs are Erd\H{o}s-R\'{e}nyi.
The communities $G'$ have the same density, disjoint set of nodes, and
are planted in non-overlapping intervals.
The background network $G$ includes nodes from all planted communities $G'$.
The edges of $G$ are generated uniformly on the timeline.
In the typical setup the length of the whole time interval $T$ is $\abs{T}=1000$ time units,
while the edges of each $G'$ are generated in intervals of length $\abs{T'}=100$ time units.
The densities of the communities and the background network vary.
The number of nodes in $G$ is set to $100$.

We test the ability of our algorithms to discover planted communities in two settings.
In the first setting (dataset family \synone) we vary the average degree of the background network
from $1$ to $6$ and fix the density of the planted $5$-cliques to $4$.
\synone allows to test the robustness against background noise.
In the second setting (dataset family \syntwo) we vary the density of planted $8$-node graphs from $2$ to $7$,
while the average degree of the background network is fixed to $2$.

\subsection{Real-world data.}
We use the following real-world datasets:
\facebook~\cite{viswanath2009evolution} is a subset of Facebook activity in the New Orleans regional community. Interactions are posts of users on each other walls.
The data covers the time period from 9.05.06 to 20.08.06.
The \twitter dataset tracks activity of Twitter users in Helsinki in year 2013.
As interactions we consider tweets that contain mentions of other users.
The \students\footnote{\url{http://toreopsahl.com/datasets/\#online\_social\_network}}
dataset logs activity in a student online network at the University of California, Irvine.
Nodes represent students and edges represent messages with ignored directions.
\enron:\footnote{\url{http://www.cs.cmu.edu/~./enron/}} is a popular dataset that contains email
communication of senior management in a large company and spans several years.

For a case study we create a hashtag network from Twitter dataset (the same tweets from users in Helsinki in year 2013): nodes represent hashtags -- there is an interaction, if two hashtags occur in the same tweet.
The timestamp of the interaction corresponds to the timestamp of the tweet. 
We denote this dataset as \twitterH.

\subsection{Optimal baseline}

A natural baseline for \segment is \segmentOpt, which combines exact dynamic programming with finding the optimal densest subgraph 
for each candidate interval. 
%
Due to the high time complexity of \segmentOpt we generate a very small dataset with $60$ timestamps,
where each timestamp contains a random graph with $3$--$6$ nodes and random density.
We vary the number of intervals $k$ and report the value of the solution
(without any post-processing) and the running time in Figure~\ref{fig:opt}.
On this toy dataset \segment is able to find near-optimal solution,
while it is significantly faster than \segmentOpt.


\begin{figure}[t]
	\begin{center}
		\setlength{\tabcolsep}{0pt}
		\newlength{\figlength}
		\setlength{\figlength}{0.45\textwidth}
		\begin{tabular}{c@{\hspace*{1mm}}*{2}{c}}
			objective function value & running time (sec.)\\
			\includegraphics[width=0.52\figlength]{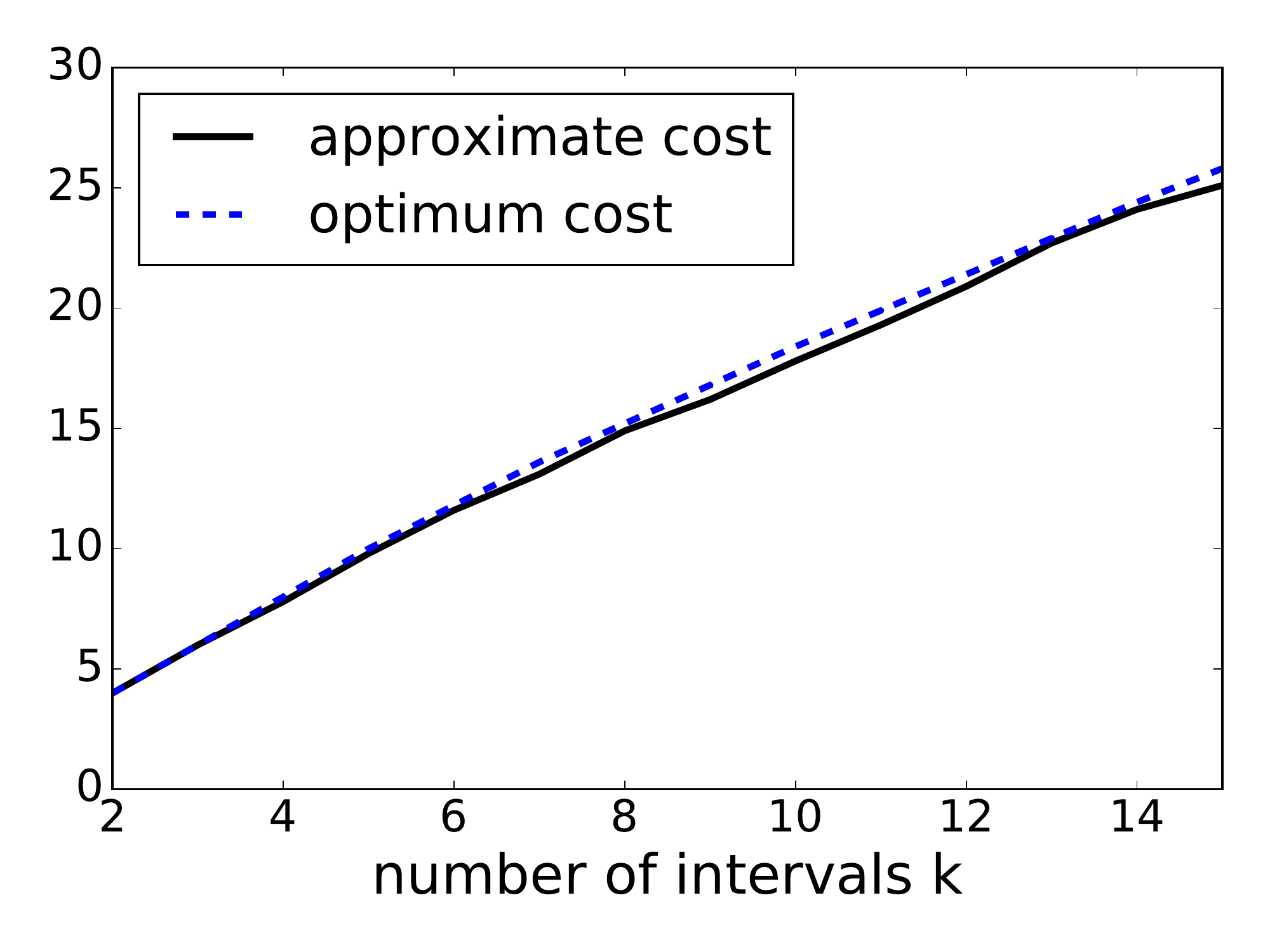}
			&\includegraphics[width=0.52\figlength]{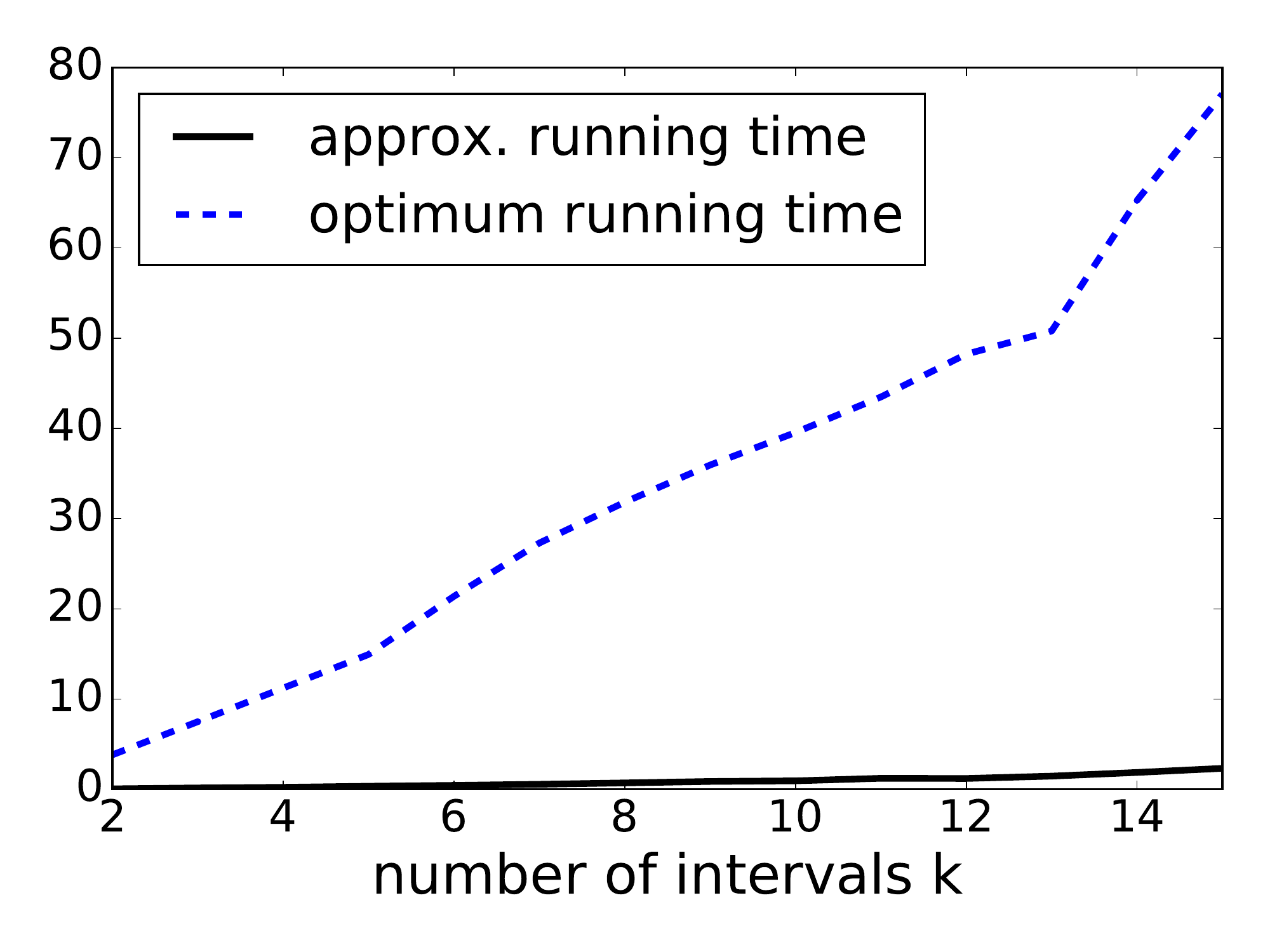}
		\end{tabular}
	\end{center}
	\caption{Comparison between optimum and approximate solutions (\segmentOpt and \segment). Approximate algorithm was run with $\epsilon_1=\epsilon_2=0.1$. Running time is in seconds.}
	\label{fig:opt}
\end{figure}

\subsection{Results on synthetic datasets}

Next, we evaluate the performance of \segment on the synthetic datasets \synone and \syntwo
by assessing how well the algorithm finds the planted subgraphs.
We report mean precision, recall, and $F$-measure, calculated with respect to the ground-truth
subgraphs.
All results are averaged over 100 independent runs.

First, Figure~\ref{fig:synth}(a) depicts the quality of the solution as a function of background noise.
Recall, that the \synone dataset contains planted 8-node subgraphs with average degree~$5$.
Precision and recall are generally high for all values of average degree in the background network.
However, precision degrades as the density of the background network increases,
as then it becomes cost-beneficial to add more nodes in the discovered densest subgraphs.

Second, Figure~\ref{fig:synth}(b) shows the quality of the solution of \segment as a function of the density in the planted subgraphs. In \syntwo the density of the background is $2$. Similarly to the previous results, the quality of the solution, especially recall, degrades much only when the density of the planted and the background network become similar.

\begin{figure}[t]
	\begin{center}
		\setlength{\figlength}{0.45\textwidth}
		\begin{tabular}{c@{\hspace*{1mm}}*{2}{c}}
			 (a) effect of background noise & (b) effect of community density\\
			\includegraphics[width=0.52\figlength]{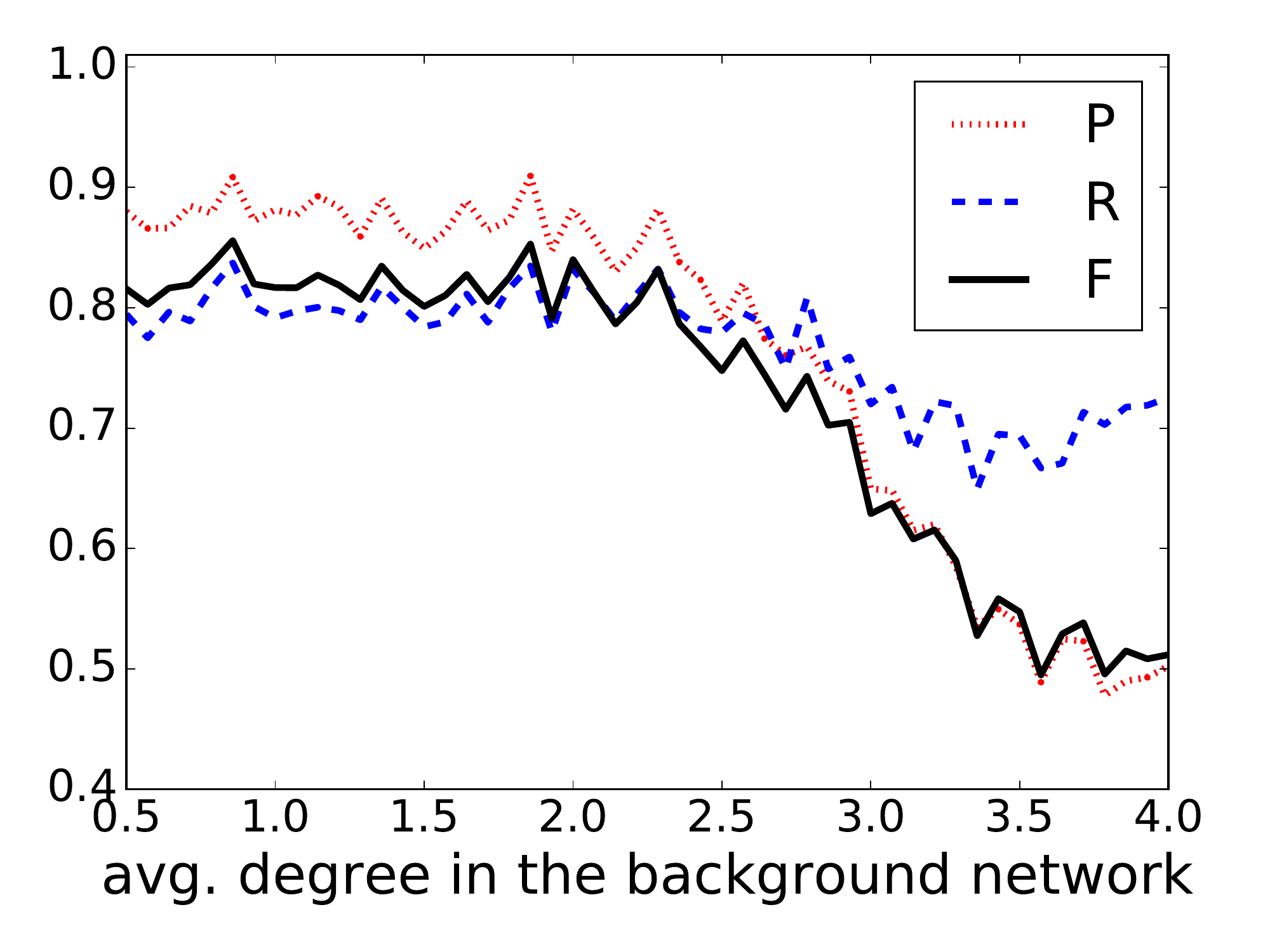}
			&\includegraphics[width=0.52\figlength]{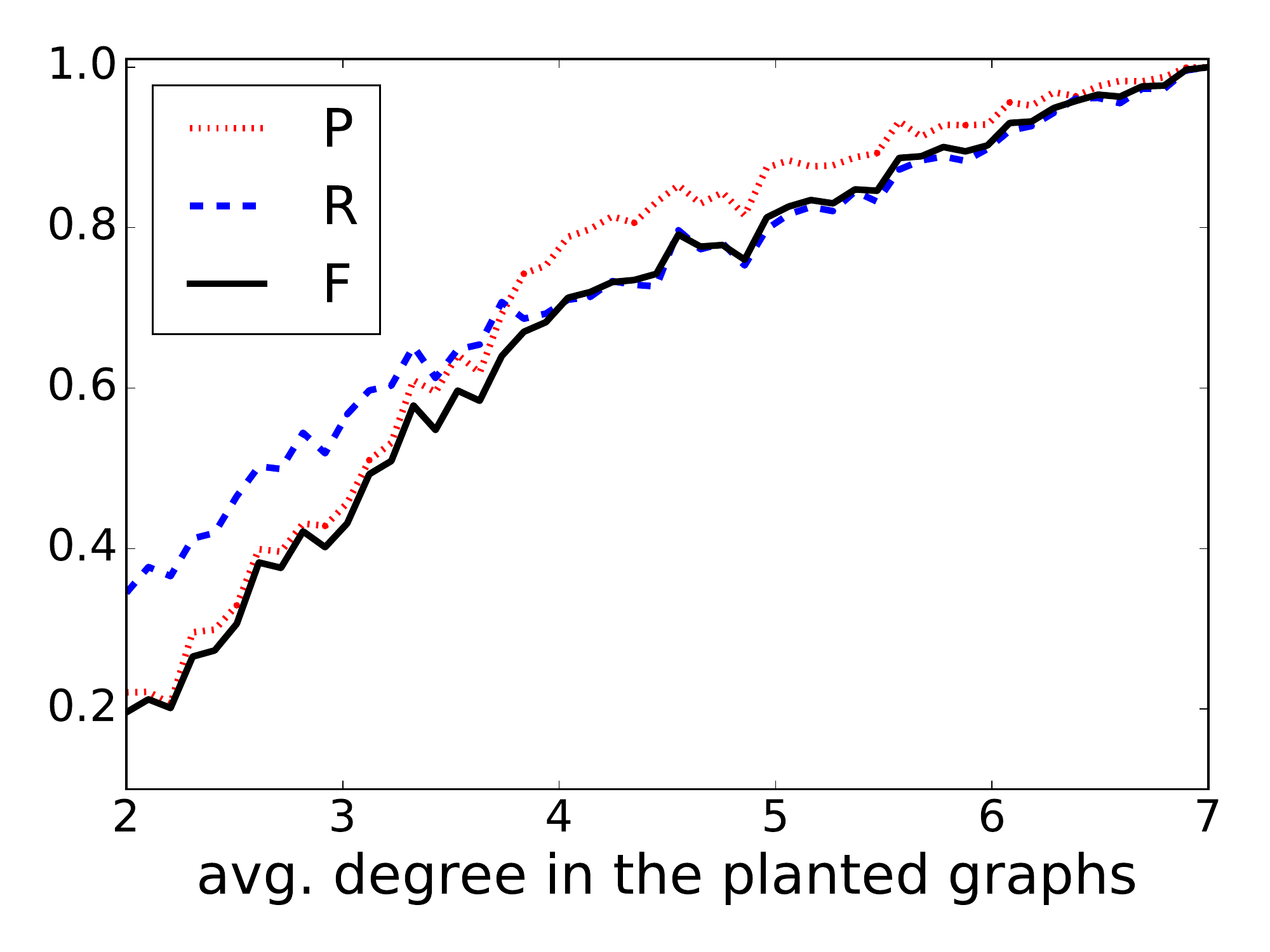}
		\end{tabular}
	\end{center}
	\caption{Precision, recall and $F$-measure on synthetic datasets. 
	For plot (a)~the community average degree is fixed to $5$ (\synone dataset), for plot (b) the background network degree is fixed to $2$ (\syntwo dataset)}.
	\label{fig:synth}
\end{figure}

\subsection{Results on real-world datasets}

As the optimal partition algorithm \segmentOpt is not scalable for real datasets,
we present comparative results of \segment with baselines \segmentDPopt and \segmentDSopt.
The \segmentDPopt algorithm performs exact dynamic programming,
but uses an approximate incremental algorithm for the densest subgraph search
(the incremental framework by Epasto et al.~\cite{epasto2015efficient}).
Vice versa, \segmentDSopt performs approximate dynamic programming
while calculating the densest subgraph optimally for each candidate interval
(by Goldberg's algorithm~\cite{goldberg1984finding}).
Note that \segmentDPopt has $2(1+\eds)^2$ approximation guarantee and
\segmentDSopt has $(1+\edp)$ approximation guarantee.
However, even these non-optimal baselines are quite slow on practice and
we use a subset of $1\,000$ interactions of \students and \enron datasets for comparative reporting.

To ensure fairness, we report the total density of the optimal densest subgraphs
in the intervals returned by the algorithms.

In Table~\ref{tab:density} we report the density of the solutions reported by
\segment, \segmentDPopt, and \segmentDSopt, as well as their running time.
We experiment with different parameters for the approximate densest-subgraph search ($\eds$)
and for approximate dynamic programming ($\edp$).

For both datasets the best solution was found by \segmentDSopt.
This is expected, as this algorithm has the best approximation factor.
The solution cost decreases as $\edp$ increases.
On the other hand, \segmentDSopt has the largest running time,
which decreases with increasing $\edp$,
but even with the largest parameter value
($\edp=2$) \segmentDSopt takes about an hour.

The \segmentDPopt algorithm typically finds the second-best solution,
however it only marginally outperforms \segment (e.g., $\eds=0.1$),
while requiring up to several orders of magnitude of higher computational time.
Naturally, the quality of the solution degrades with increasing $\eds$.

The solution quality degrades with increasing the approximation parameters for all algorithms.
However, the degradation is not as dramatic as the worst case bound suggests, and
using such an approximation parameter offers significant speed-up.
\segment provides the fastest estimates of a good quality
for a wide range of approximation parameters.
Note that \segment is more sensitive to the changes in the quality of the densest
subgraph search regulated by $\eds$.

\begin{table*}[t]
	\begin{center}
		\caption{Comparison with \segmentDPopt and \segmentDSopt baselines.}\label{tab:density}					
		\begin{tabular}{l @{\hspace{0.5mm}} c @{\hspace{1mm}} c}
			\toprule
			 Dataset& Community density & Running time (sec.)\\
			 \cmidrule{1-2}
			 \cmidrule{3-3}
			
			\students1000 &
			\begin{tabular}{ll|llll|l}
				&&\multicolumn{4}{c|}{$\eds$} &\\
				&\segment & 0.01 & 0.1& 1& 2 &\segmentDSopt \\\midrule
				
				\multirow{4}{*}{$\edp$}
				&0.01 &4.24&4.24 &3.82 &3.82 & 6.30\\
				&0.1 & 4.24&4.24 &3.82 &3.82 &6.22 \\
				&1 &  4.24&4.24 &3.82 &3.82  & 5.76\\
				&2 &  4.24&4.24 &3.82 &3.82  &5.61 \\\midrule
				&\segmentDPopt & 5.73 &5.73 & 3.82&3.82  & \\
			\end{tabular} &
			\begin{tabular}{ll|llll|l}
				&&\multicolumn{4}{c|}{$\eds$} &\\
				&\segment & 0.01 & 0.1& 1& 2 &\segmentDSopt \\\midrule
				
				\multirow{4}{*}{$\edp$}&
				0.01 &0.62 & 0.62& 0.63&0.64 &23678 \\
				&0.1 &0.23 & 0.23&0.24 &0.23 & 11877\\
				&1 & 0.13& 0.26&0.13 &0.13 & 3394\\
				&2 &0.36 & 0.20& 0.20&0.36 & 3769\\\midrule
				&\segmentDPopt &162 &43.5 & 29.5&  29.5& \\
			\end{tabular}
			
			\\\midrule
			\enron1000 &
			\begin{tabular}{ll|llll|l}
				&&\multicolumn{4}{c|}{$\eds$} &\\
				&\segment & 0.01 & 0.1& 1& 2 &\segmentDSopt \\\midrule
				
				\multirow{4}{*}{$\edp$}
				&0.01 &10.4 &10.4 & 10.0&10.5 & 11.3\\
				&0.1 &10.3 & 10.4& 10.0&10.3 & 11.0\\
				&1 & 9.54&9.54 &8.80 &9.83 & 11.0\\
				&2 &7.34 & 7.34& 7.34&7.34 & 10.8\\\midrule
				&\segmentDPopt &10.5 &11.0 &10.4 &10.4  & \\				
			\end{tabular} &
			\begin{tabular}{ll|llll|l}
				&&\multicolumn{4}{c|}{$\eds$} &\\
				&\segment & 0.01 & 0.1& 1& 2 &\segmentDSopt \\\midrule				
				\multirow{4}{*}{$\edp$}
				&0.01 & 56.4& 55.5& 42.3&31.8 & 25788\\
				&0.1 &3.02 &2.85 & 2.07 & 1.70& 16070\\
				&1 & 0.43& 0.44&0.29 &0.28 & 7834\\
				&2 & 0.22& 0.22& 0.23 &0.23 & 3469\\\midrule
				&\segmentDPopt & 1654& 61.15& 17.82&  6.07& \\
			\end{tabular}			
			\\\bottomrule
		\end{tabular}
	\end{center}
\end{table*}

\subsection{Running time and scalability}

Figure~\ref{fig:epsilons} shows running time of \segment
as a function of the approximation parameters $\eds$ and $\edp$.
The figure confirms the theory, that is,
$\eds$ has significant impact on the running time,
while the algorithm scales very well with~\edp.

We demonstrate scalability in Figure~\ref{fig:length},
plotting the running time for increasing number of interactions,
for \facebook and \twitter datasets.
Recall that the theoretical running time is $\bigO(k^2 m\log n)$,
where $n$ is the number of nodes and $m$ the number of interactions.
In practice, the running time grows fast for the first thousand interactions and
then saturates to linear dependence.
This happens because in the beginning of the network history the number of nodes grows fast.
In addition, new, denser than previously seen, subgraphs are more likely to occur.
Thus, the approximate densest-subgraph subroutine has to be computed more often.
Furthermore, the number of intervals $k$ contributes to running time as expected.

\begin{figure}[t]
	\begin{center}		
		\begin{tabular}{c@{\hspace{1mm}}c@{\hspace{-1mm}}c}
			& \students &\twitter\\
			\rotatebox{90}{\hspace*{0.3cm}\small{Running time} (sec)}&\includegraphics[width=0.45\columnwidth]{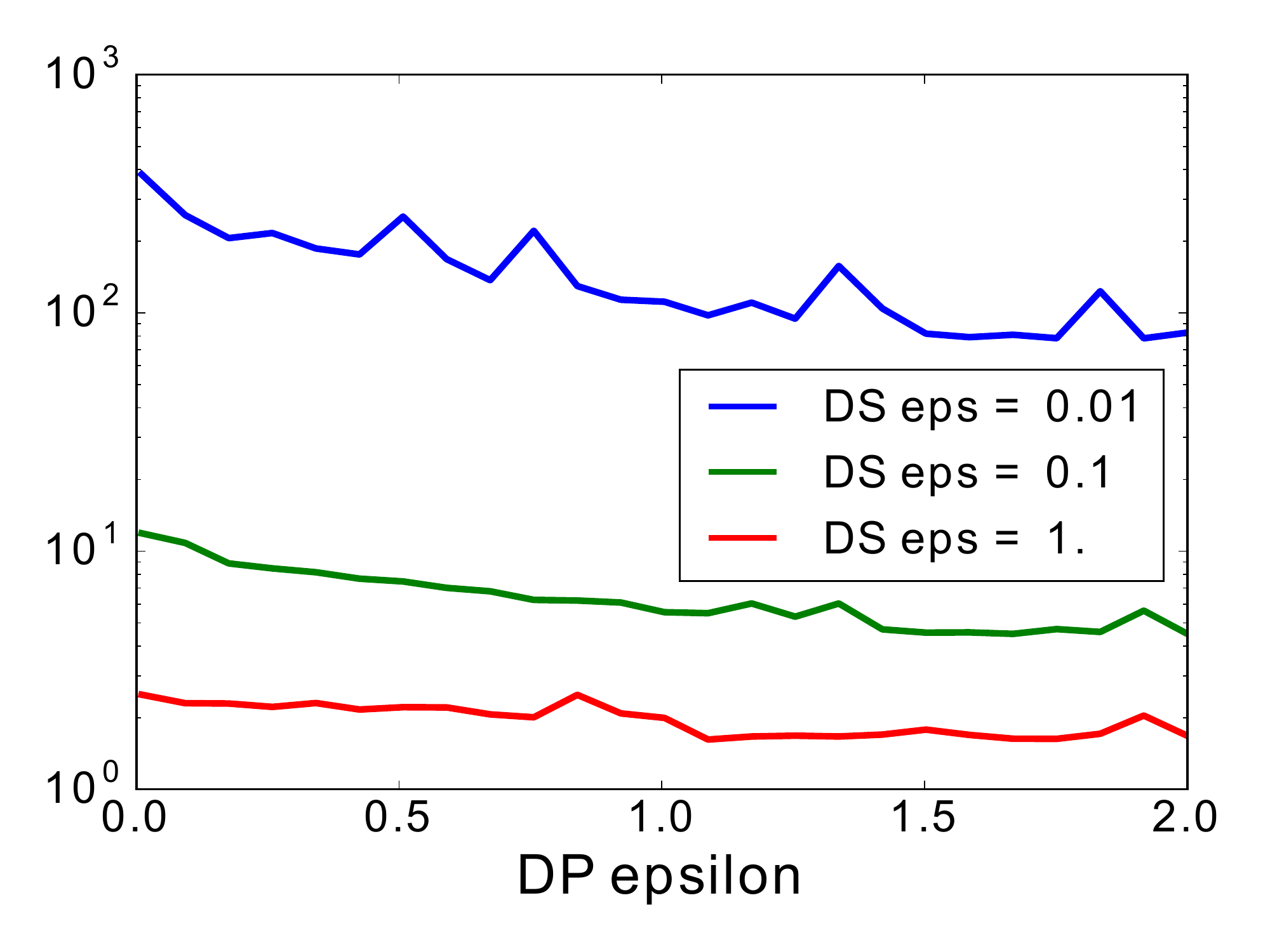}
			&\includegraphics[width=0.49\columnwidth]{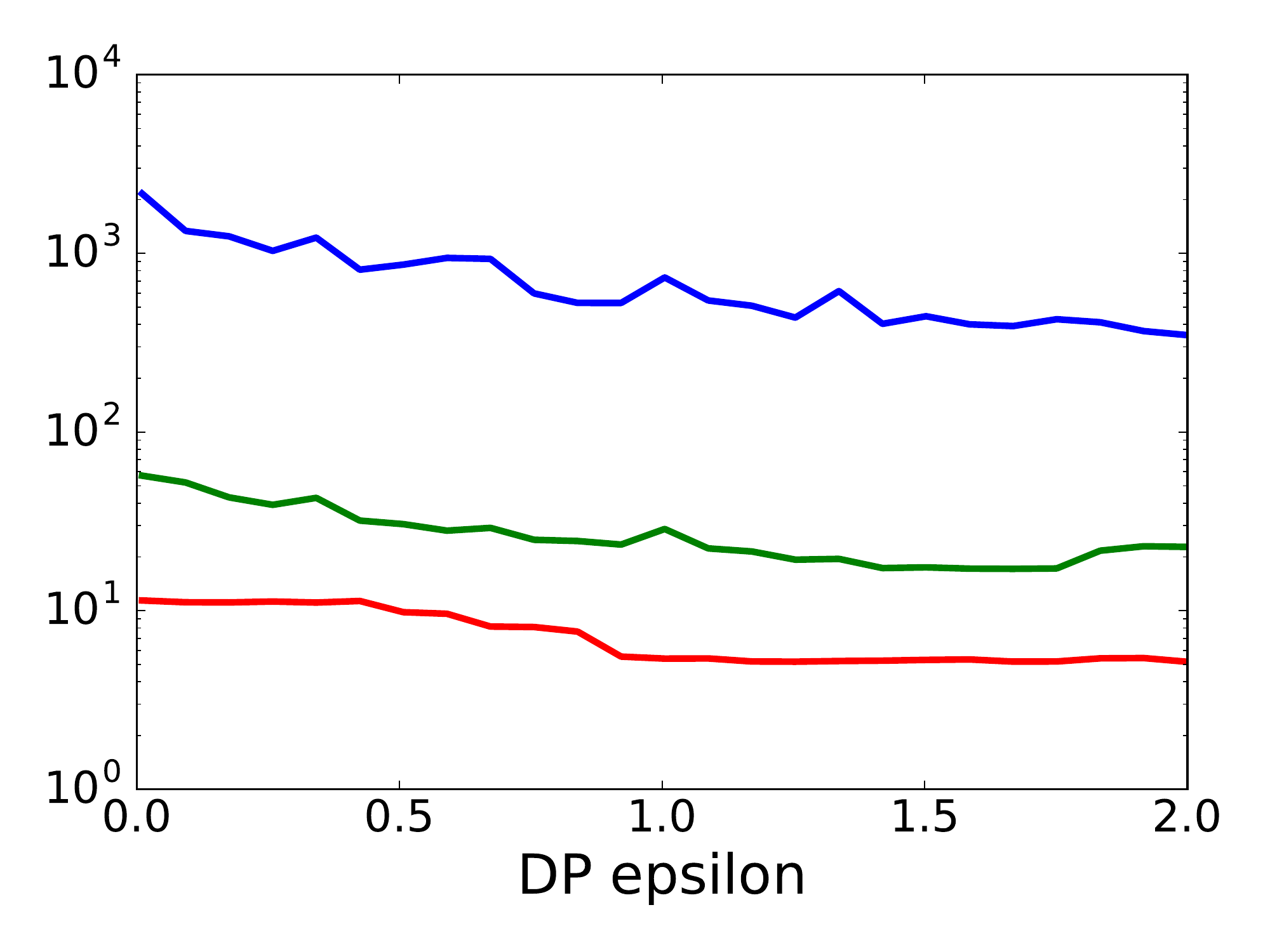}
		\end{tabular}
	\end{center}
	\caption{Effect of different approximation parameters in \segment. $k=20$.}
	\label{fig:epsilons}
\end{figure}

\begin{figure}[t]
	\begin{center}		
		\begin{tabular}{c@{\hspace{0.5mm}}c@{\hspace{-2mm}}c}
			& \facebook &\twitter\\
			\rotatebox{90}{\hspace*{0.3cm}\small{Running time} (sec)}&\includegraphics[width=0.49\columnwidth]{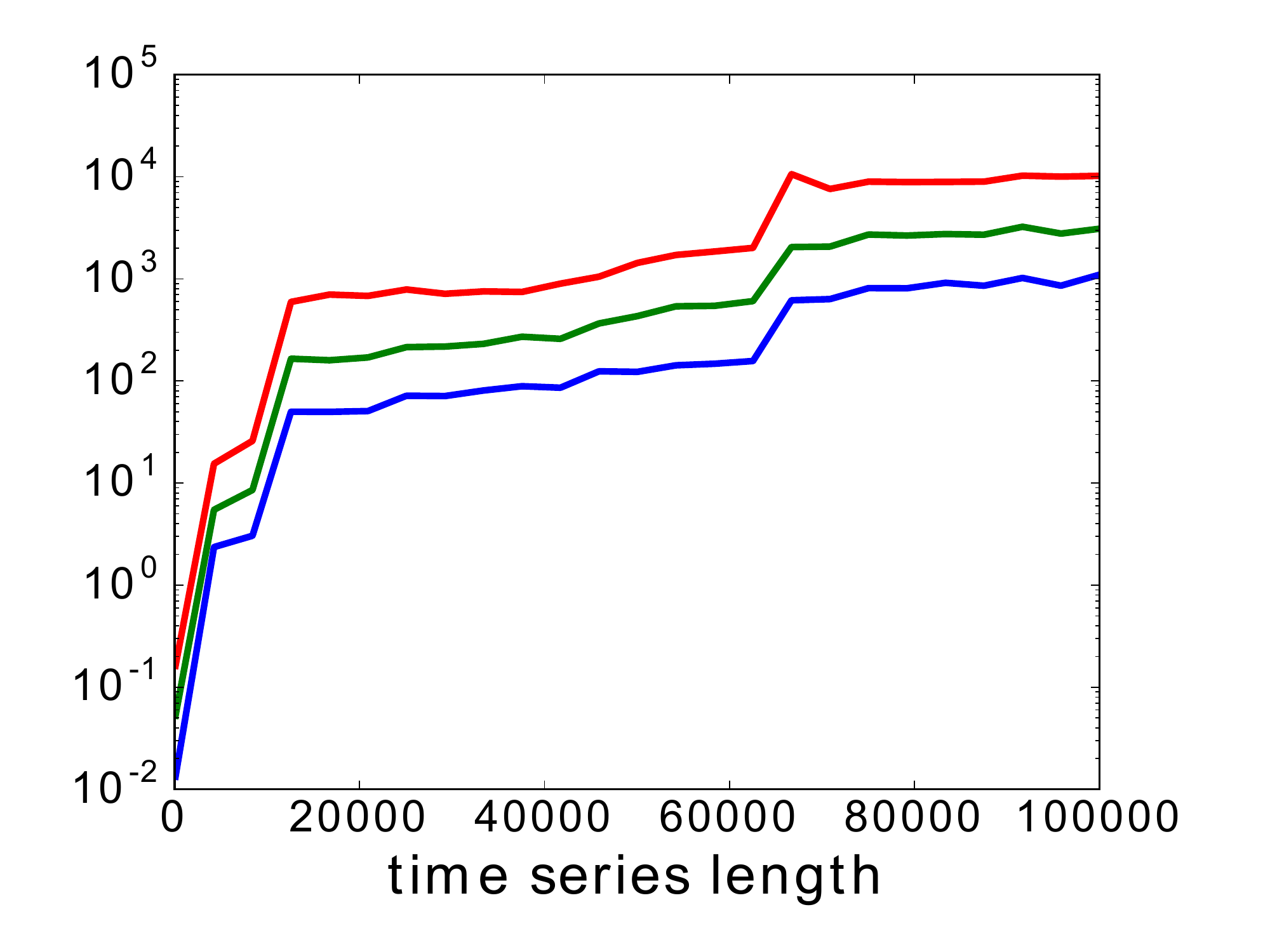}
			&\includegraphics[width=0.49\columnwidth]{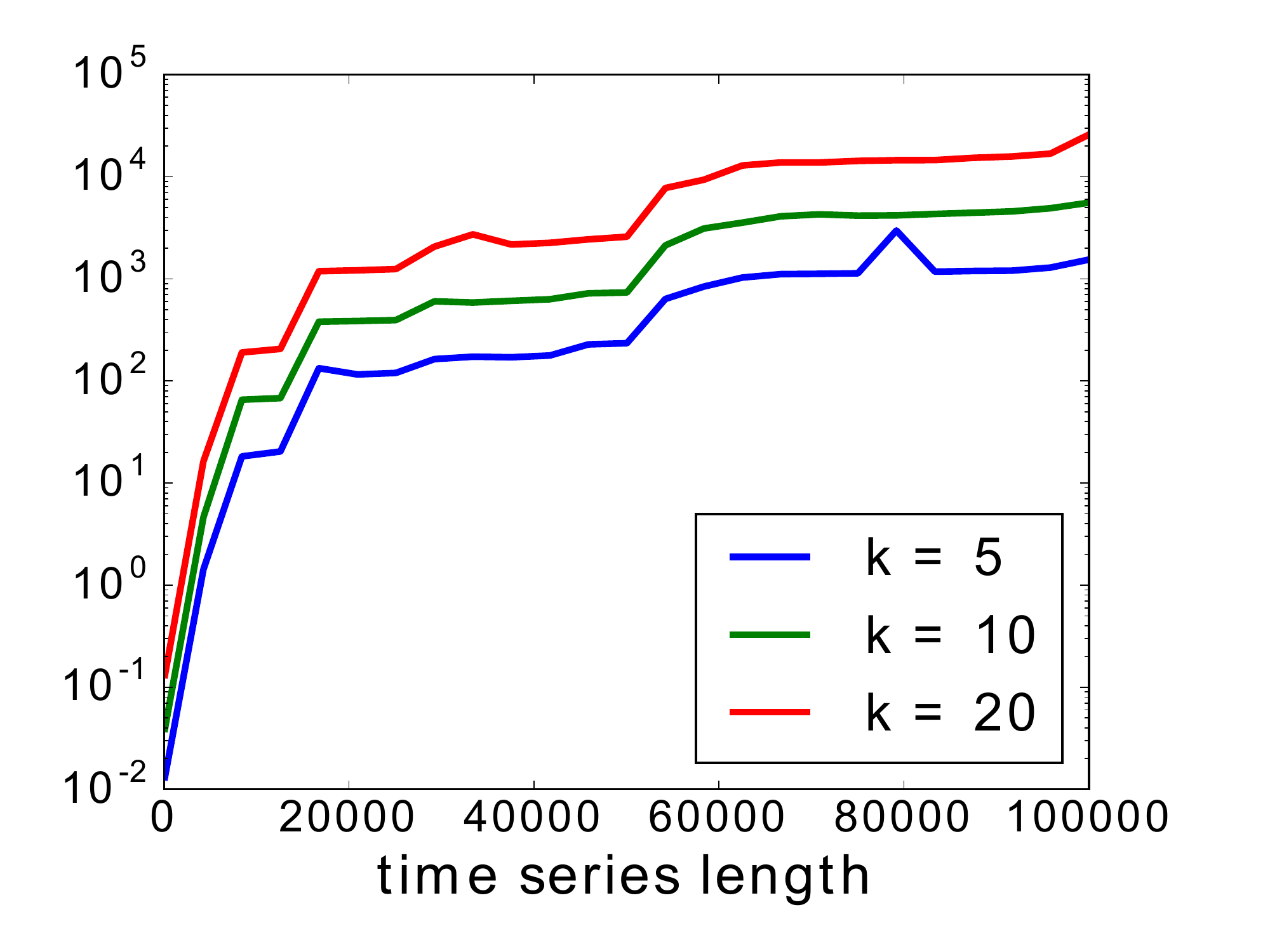}
		\end{tabular}
	\end{center}
	\caption{Scalability testing with $\eds=\edp=0.1$.}
	\label{fig:length}
\end{figure}

\subsection{Subgraphs with larger node coverage --- static graphs}

Next we evaluate \staticGreedy.
To measure coverage, we simply count the number of distinct nodes in the output subgraphs.
We use the 10K first interactions of \students dataset, set $k=20$,
and test different values of $\lambda$.
Figure~\ref{fig:static} shows the density and the pairwise Jaccard similarity
of the node sets of the retrieved subgraphs.
The subgraphs are shown in the order they are discovered.
Smaller values of $\lambda$ give larger density,
and larger values of $\lambda$ give more cover.
We observe that, for all values of $\lambda$,
in the beginning \staticGreedy returns diverse and dense subgraphs,
but soon after it returns identical graphs.
We speculate that the algorithm finds all dense subgraphs that exist in the dataset.
Regarding setting $\lambda$,  we observe that
$\lambda=0.002$ offers a good trade-off in finding
subgraphs of high density and moderate overlap.


\begin{figure*}[t]
	\begin{center}		
		\setlength{\figlength}{0.45\textwidth}
		\begin{tabular}{c@{\hspace*{1mm}}*{4}{c}}
		 	$\lambda=0.001$ & $\lambda=0.002$ & $\lambda=0.003$ &\\
			\includegraphics[width=0.52\figlength]{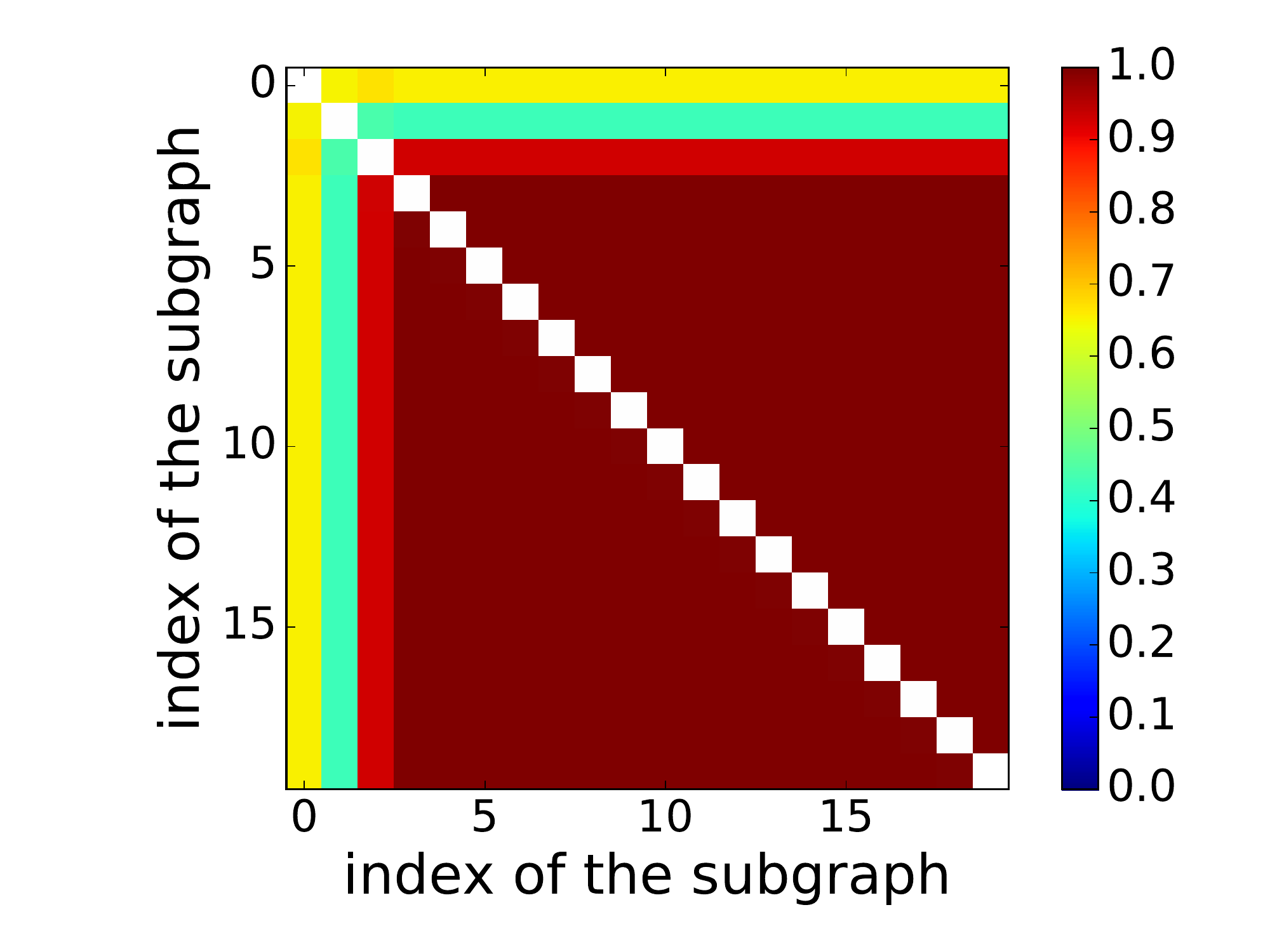}
			&\includegraphics[width=0.52\figlength]{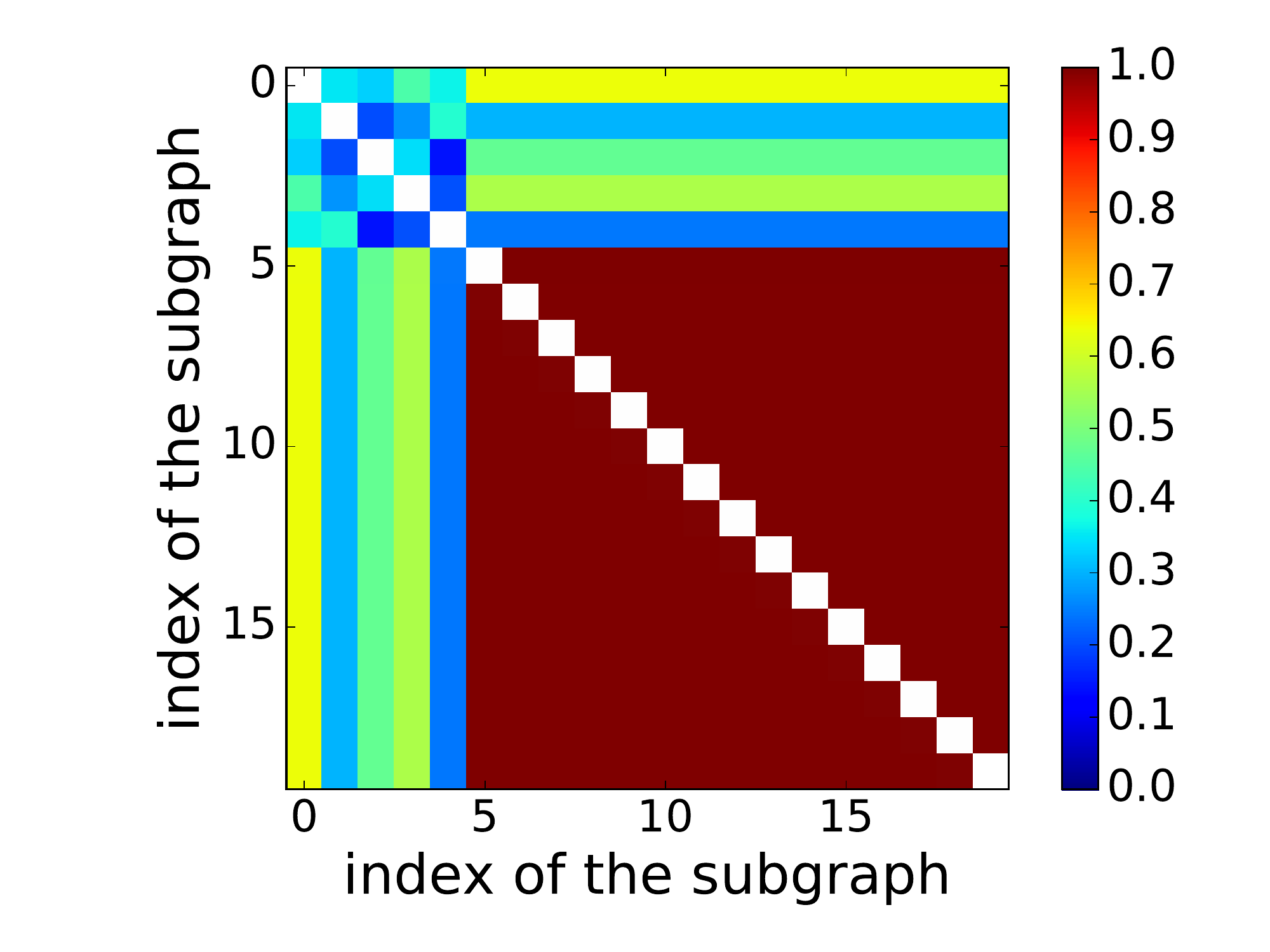}
			&\includegraphics[width=0.52\figlength]{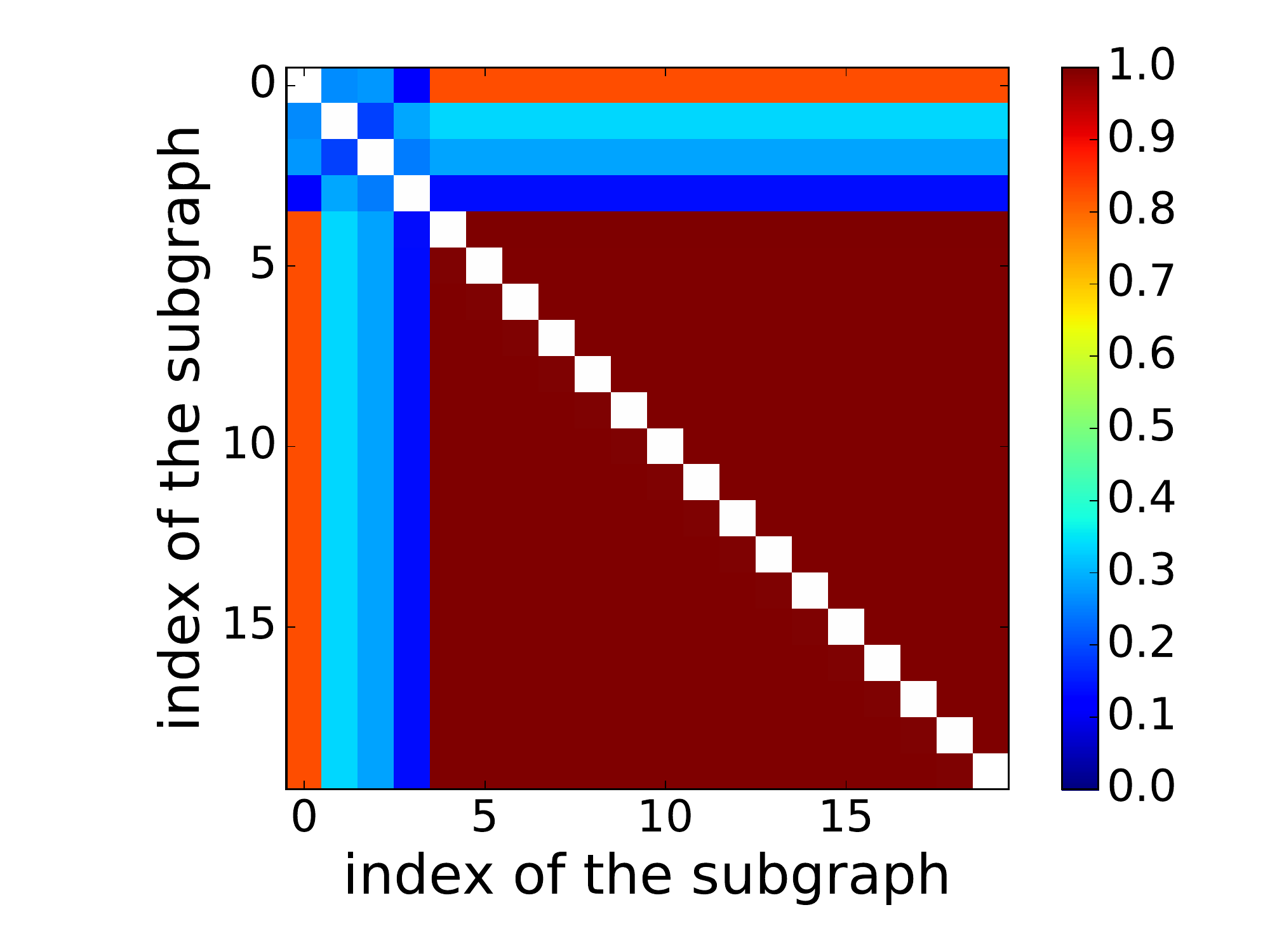}
			&\includegraphics[width=0.52\figlength]{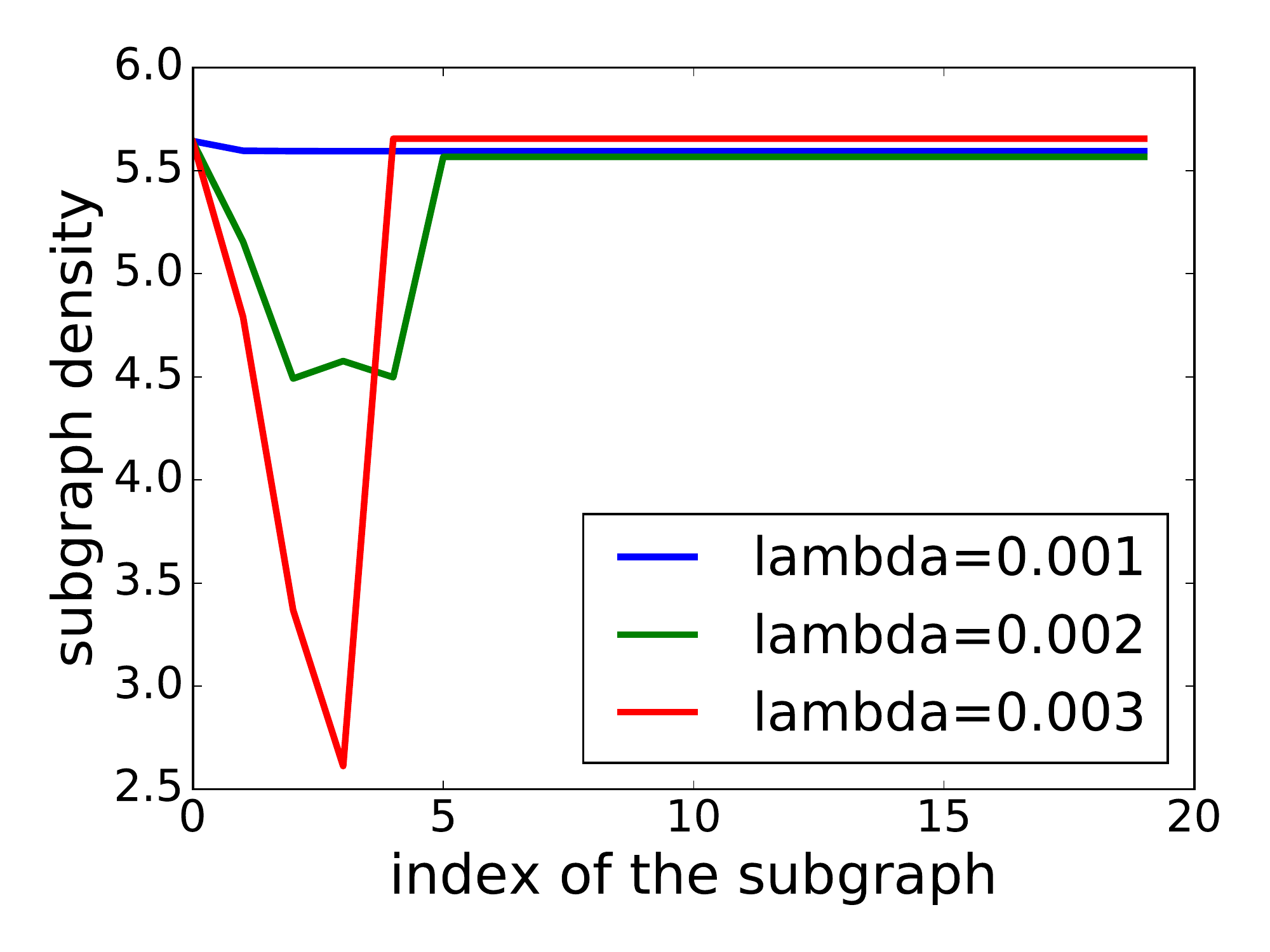}\\
		\end{tabular}
	\end{center}
	\caption{Pairwise similarities (3 heatmap plots on the left) and densities (right plot) of subgraphs returned by \staticGreedy.}
	\label{fig:static}
\end{figure*}

\subsection{Subgraphs with larger node coverage --- dynamic graphs}

Finally we evaluate the performance of \segmentCover algorithm.
We vary the parameter $\lambda$ and compare different characteristics of the solution,
with the solution returned by \segment.
For different values of $\lambda$,
Table~\ref{tab:overlap} shows average density, total number of covered nodes,
average size of the subgraphs, and average pairwise Jaccard similarity.
Although \segmentCover does not have an approximation guarantee,
for small values of $\lambda$ it finds subgraphs of the density close to \segment.
Similarly to the static case, $\lambda$ provides an efficient trade-off between density and coverage.


\begin{table*}[t]
	\begin{center}
		\caption{Results of \segmentCover with $k=5$ and $\eds = \edp = 0.1$.}\label{tab:overlap}
		\begin{tabular}{ll ll l ll l ll l ll}
			\toprule &&
			\multicolumn{2}{c}{Density} &&
			\multicolumn{2}{c}{Cover} &&
			\multicolumn{2}{c}{Size} &&
			\multicolumn{2}{c}{JSim} \\
			\cmidrule{3-4}
			\cmidrule{6-7}
			\cmidrule{9-10}
			\cmidrule{12-13}
			
			Dataset & $\lambda$ & \segmentCover &\segment && \segmentCover &\segment && \segmentCover &\segment && \segmentCover &\segment\\ 
			\midrule
			
			\students & 1e-6 & 10.690 & 11.151 && 136 & 130 && 48.75 & 37.6 && 0.1449 & 0.0951\\ 
			&1e-5 & 7.0869 & 11.151 && 813 & 130 && 261.0 & 37.6 && 0.0788 & 0.095\\ 
			&1e-4 & 5.0273 & 11.151 && 889 & 130 && 286.0 & 37.6 && 0.0910 & 0.0951\\ 
			\midrule	
				
			\enron & 1e-6 & 19.995 & 19.871 && 38 & 37 && 16.0 & 16.2 && 0.3619 & 0.3851\\ 
			&1e-5 & 19.962 & 19.871 && 40 & 37 && 17.0 & 16.2 && 0.3660 & 0.3851\\ 
			&1e-4 & 6.5684 & 19.871 && 1144 & 37 && 288.8 & 16.2 && 0.0808 & 0.3851\\ 
			\midrule
			
			\facebook &1e-8 & 5.3714 & 5.3933 && 83 & 120 && 22.75 & 27.6 && 0.0185 & 0.0163\\ 
			&1e-7 & 4.2749 & 5.3933 && 3470 & 120 && 882.0 & 27.6 && 0.0027 & 0.0163\\ 
			&1e-6 & 3.2673 & 5.3933 && 4100 & 120 && 1228.75 & 27.6 && 0.0335 & 0.0163\\ 
			\midrule
			
			\twitter & 1e-7 & 9.9970 & 10.138 && 128 & 152 && 44.25 & 54.0 && 0.1590 & 0.1673\\ 
			&1e-6 & 6.5500 & 10.138 && 3808 & 152 && 1061.75 & 54.0 && 0.0837 & 0.1673\\ 
			&1e-5 & 3.5389 & 10.138 && 4604 & 152 && 1379.0 & 54.0 && 0.0773 & 0.1673\\ 
			\bottomrule
		\end{tabular}
	\end{center}
\end{table*}

\section{Case study}
\label{sec:casestudy}

We present a case study using graphs of co-occurring hashtags from Twitter messages in the Helsinki region.
We create two subsets of \twitterH dataset:
one covering all tweets in November 2013 and another in December 2013.
Figure~\ref{fig:twitter} shows the dense subgraphs discovered by the \segment algorithm on these datasets,
with $k=4$ and $\eds = \edp = 0.1$.

For the November dataset,
\segment creates a small 1-day interval in the beginning and then splits the rest time almost evenly.
This first interval includes the nodes \texttt{\small movember},
\texttt{\small liiga},
\texttt{\small halloween}, and
\texttt{\small digiexpo},
which cover a broad range of global
(e.g., movember and Halloween)
and local events
(e.g., game-industry event DigiExpo and Finnish ice-hockey league).
The next interval is represented by a large variety of well-connected tags related to
\texttt{\small mtv} and media,
corresponding to the MTV Europe Music Awards'13 on November~10.
There are also other ice hockey-related tags, e.g.,  \texttt{\small leijonat},
and Father's Day tags, e.g., \texttt{\small isänpäivä}, which was on November~13.
The third interval is mostly represented by Slush-related tags;
Slush is the annual large startup and tech event in Helsinki.
The last interval is completely dedicated to ice-hockey with many team names.

There are three major public holidays in December:
Finland's Independence Day on December 6,
Christmas on December 25, and
New Year's Eve on December 31.
\segment allocates one interval for Christmas and New Year from December 21 to 31.
Ice hockey is also represented in this interval,
as well as in the third interval.
Remarkably, the Independence Day holiday is split into 2 intervals.
The first one is from December 1 to December 6, 3:30pm, and the corresponding graph has two clusters:
the first one containing general holidays-related tags and the second one is focused on Independence Day President's reception.
This is a large event that starts on December 6, 6pm,
is broadcasted live, and is discussed in media for the following days.
The second interval for December 6-9 is a truthful representation of this event.

\begin{figure*}[t]
\vspace{2mm}
	\begin{center}		
		\setlength{\figlength}{0.45\textwidth}		
		\begin{tabular}{c@{\hspace*{1mm}}*{5}{c}}
			\rotatebox{90}{\hspace*{0.3cm} November 2013}
			&\scalebox{0.65}{\input{figures/Nov/filtered_2013_11_1_md2}}
			&
			\scalebox{0.65}{\input{figures/Nov/filtered_2013_11_2_md1}}
			&
			\scalebox{0.65}{\input{figures/Nov/filtered_2013_11_3_md3}}
			&
			\scalebox{0.65}{\input{figures/Nov/filtered_2013_11_4_md1}}
			\\
			& 01.11 03:31 -- 02.11 11:00
			& 02.11 11:01 -- 11.11 12:53
			& 11.11 12:54 -- 21.11 06:12
			& 21.11 06:18 -- 30.11 22:40\\
			\rotatebox{90}{\hspace*{0.3cm} December 2013}
			&
			\scalebox{0.65}{\input{figures/Dec/filtered_2013_12_1_md1}}
			&
			\scalebox{0.65}{\input{figures/Dec/filtered_2013_12_2_md1}}
			&
			\scalebox{0.65}{\input{figures/Dec/filtered_2013_12_3_md1}}
			&
			\scalebox{0.65}{\input{figures/Dec/filtered_2013_12_4_md1}}
			\\
			& 01.12 00:07 -- 06.12 15:29
			& 06.12 15:29 -- 09.12 17:58
			& 09.12 18:36 -- 21.12 12:23
			& 21.12 12:33 -- 31.12 23:41\\	
		\end{tabular}
	\end{center}
	\caption{Subgraphs, discovered in the network of Twitter hashtags \twitterH by \segment algorithm with $k=4$, $\eds = \edp = 0.1$.}
	\label{fig:twitter}
\vspace{2mm}
\end{figure*}
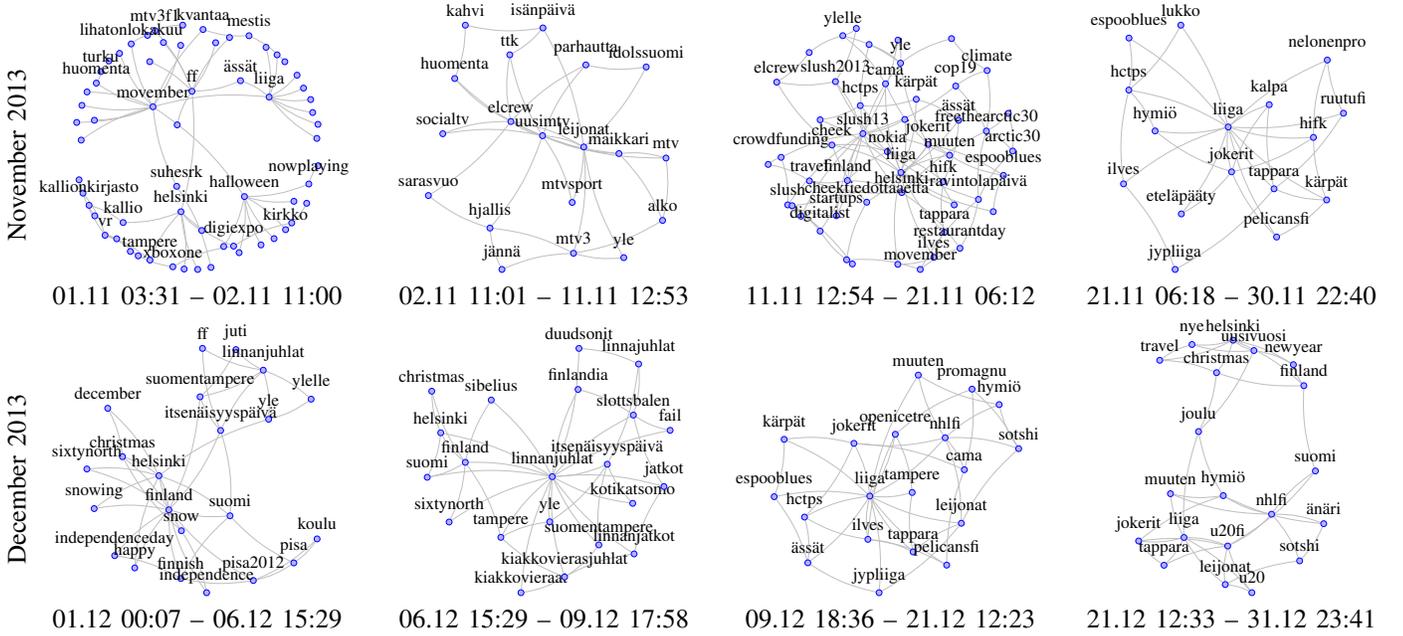

%% file: figures/Nov/filtered_2013_11_1_md2.tex
\begin{tikzpicture}[scale = 1] 
\tikzstyle{node_style} = [circle,draw=blue,fill=blue!30!,inner sep = 1.2pt] 
\tikzstyle{edge_style} = [gray!50, bend left = 10]
\node[node_style] (v0) at (2.13844396529,1.21209594267) {};
\node[node_style] (v1) at (2.28511094076,0.977127537855) {};
\node[node_style] (v2) at (-2.40784297312,0.147354155738) {};
\node[node_style] (v3) at (-2.13010535346,-1.40321125995) {};
\node[node_style] (v4) at (-2.37158094057,-0.937911023022) {};
\node[node_style] (v5) at (-0.373611984516,2.08685614168) {};
\node[node_style, label={xboxone} ] (v6) at (-0.533890370328,-2.44592075949) {};
\node[node_style] (v7) at (2.20326862737,-1.1457255056) {};
\node[node_style, label={vr} ] (v8) at (-1.92,-1.8) {};
\node[node_style] (v9) at (-0.444613263825,0.459339812843) {};
\node[node_style] (v10) at (1.75637702898,1.75167274815) {};
\node[node_style] (v11) at (-2.45204542682,-0.665012945013) {};
\node[node_style] (v12) at (-1.84223702367,1.76705203287) {};
\node[node_style, label={kirkko} ] (v13) at (1.77234932174,-1.68602428724) {};
\node[node_style] (v14) at (0.342337378411,2.14021620711) {};
\node[node_style] (v15) at (-1.40232234685,-2.13105621771) {};
\node[node_style, label={turku} ] (v16) at (-2.00577227463,1.5498899276) {};
\node[node_style] (v17) at (2.43902477386,-0.3767560278) {};
\node[node_style, label={kvantaa} ] (v18) at (0.0837334493188,2.41150249134) {};
\node[node_style] (v19) at (-1.00531182156,1.75294097671) {};
\node[node_style, label={suhesrk} ] (v20) at (-0.454866339231,-0.798588618259) {};
\node[node_style] (v21) at (2.32593690788,0.711801150531) {};
\node[node_style] (v22) at (1.59804265953,1.89246726692) {};
\node[node_style, label={mestis} ] (v23) at (1.02403363079,2.28397971147) {};
\node[node_style] (v24) at (1.37263182432,2.04697788661) {};
\node[node_style, label={tampere} ] (v25) at (-1.00189085394,-2.30879586766) {};
\node[node_style] (v26) at (-0.72622419089,2.14478937944) {};
\node[node_style] (v27) at (-2.13816588605,0.553882947362) {};
\node[node_style] (v29) at (1.9424384321,-1.10521515054) {};
\node[node_style, label={kallionkirjasto} ] (v30) at (-2.24682566128,-1.18680043605) {};
\node[node_style, label={ässät} ] (v31) at (0.850639546229,1.36284196863) {};
\node[node_style, label={digiexpo} ] (v32) at (0.709389248581,-2.02272404858) {};
\node[node_style] (v33) at (1.26976822382,-1.99753861206) {};
\node[node_style] (v34) at (1.53621577412,-1.86952190626) {};
\node[node_style, label={halloween} ] (v35) at (0.931199740557,-1.00834270575) {};
\node[node_style] (v36) at (-2.39328352251,0.858877121706) {};
\node[node_style] (v37) at (1.89773036586,-1.54862717596) {};
\node[node_style] (v38) at (0.0569993131511,-1.70277246642) {};
\node[node_style] (v39) at (2.43441106251,0.469903696198) {};
\node[node_style] (v40) at (-0.331312914093,2.5) {};
\node[node_style, label={mtv3f1} ] (v41) at (-0.898461358234,2.38669799093) {};
\node[node_style] (v42) at (0.821710711629,-2.15626655389) {};
\node[node_style, label={ff} ] (v43) at (-0.141721009712,1.14679543155) {};
\node[node_style] (v44) at (-2.5,0.509015919057) {};
\node[node_style] (v45) at (0.625666518197,2.24805292724) {};
\node[node_style, label={liiga} ] (v46) at (1.43510868607,1.02891786714) {};
\node[node_style] (v47) at (0.245727759871,-2.45855308476) {};
\node[node_style, label={helsinki} ] (v48) at (-0.367664754179,-1.31848545486) {};
\node[node_style] (v49) at (-0.302334807002,-2.49118143029) {};
\node[node_style] (v50) at (2.00621917441,1.47647160661) {};
\node[node_style] (v51) at (-1.24154116994,-2.22184304776) {};
\node[node_style] (v52) at (2.41175578019,0.182263504719) {};
\node[node_style, label={nowplaying} ] (v53) at (2.24714228331,-0.751046211979) {};
\node[node_style, label={movember} ] (v54) at (-0.939575224961,0.829690271457) {};
\node[node_style] (v55) at (-0.0250643306974,-2.5) {};
\node[node_style] (v56) at (-1.05193149454,2.29184982552) {};
\node[node_style] (v57) at (-2.28841785672,1.13355041341) {};
\node[node_style, label={kallio} ] (v58) at (-1.55058454254,-1.54016241997) {};
\node[node_style] (v59) at (-1.68185526823,-1.87300525012) {};
\node[node_style] (v60) at (0.506675721257,-2.02921658812) {};
\node[node_style] (v61) at (-1.62485150721,1.92329503533) {};
\node[node_style, label={huomenta} ] (v62) at (-2.09924411899,1.31355434324) {};
\node[node_style, label={lihatonlokakuu} ] (v28) at (-1.38548981897,2.11552083794) {};
\begin{pgfonlayer}{bg}
\draw[edge_style]  (v0) edge (v46);
\draw[edge_style]  (v1) edge (v46);
\draw[edge_style]  (v53) edge (v17);
\draw[edge_style]  (v53) edge (v35);
\draw[edge_style]  (v2) edge (v54);
\draw[edge_style]  (v3) edge (v30);
\draw[edge_style]  (v4) edge (v8);
\draw[edge_style]  (v5) edge (v54);
\draw[edge_style]  (v6) edge (v32);
\draw[edge_style]  (v6) edge (v15);
\draw[edge_style]  (v7) edge (v13);
\draw[edge_style]  (v8) edge (v11);
\draw[edge_style]  (v8) edge (v51);
\draw[edge_style]  (v8) edge (v59);
\draw[edge_style]  (v8) edge (v25);
\draw[edge_style]  (v9) edge (v43);
\draw[edge_style]  (v12) edge (v16);
\draw[edge_style]  (v14) edge (v43);
\draw[edge_style]  (v16) edge (v62);
\draw[edge_style]  (v19) edge (v43);
\draw[edge_style]  (v20) edge (v54);
\draw[edge_style]  (v20) edge (v47);
\draw[edge_style]  (v21) edge (v46);
\draw[edge_style]  (v22) edge (v46);
\draw[edge_style]  (v23) edge (v18);
\draw[edge_style]  (v23) edge (v10);
\draw[edge_style]  (v24) edge (v46);
\draw[edge_style]  (v25) edge (v48);
\draw[edge_style]  (v57) edge (v54);
\draw[edge_style]  (v26) edge (v43);
\draw[edge_style]  (v54) edge (v18);
\draw[edge_style]  (v54) edge (v35);
\draw[edge_style]  (v54) edge (v27);
\draw[edge_style]  (v54) edge (v62);
\draw[edge_style]  (v54) edge (v28);
\draw[edge_style]  (v54) edge (v44);
\draw[edge_style]  (v54) edge (v36);
\draw[edge_style]  (v54) edge (v61);
\draw[edge_style]  (v54) edge (v46);
\draw[edge_style]  (v28) edge (v56);
\draw[edge_style]  (v29) edge (v35);
\draw[edge_style]  (v30) edge (v58);
\draw[edge_style]  (v31) edge (v46);
\draw[edge_style]  (v31) edge (v43);
\draw[edge_style]  (v32) edge (v37);
\draw[edge_style]  (v32) edge (v48);
\draw[edge_style]  (v33) edge (v35);
\draw[edge_style]  (v34) edge (v35);
\draw[edge_style]  (v41) edge (v40);
\draw[edge_style]  (v41) edge (v43);
\draw[edge_style]  (v35) edge (v60);
\draw[edge_style]  (v35) edge (v42);
\draw[edge_style]  (v35) edge (v13);
\draw[edge_style]  (v38) edge (v48);
\draw[edge_style]  (v39) edge (v46);
\draw[edge_style]  (v43) edge (v48);
\draw[edge_style]  (v43) edge (v45);
\draw[edge_style]  (v55) edge (v48);
\draw[edge_style]  (v48) edge (v49);
\draw[edge_style]  (v48) edge (v58);
\draw[edge_style]  (v50) edge (v46);
\draw[edge_style]  (v52) edge (v46);
\end{pgfonlayer}
\end{tikzpicture}

%% file: figures/Nov/filtered_2013_11_2_md1.tex
\begin{tikzpicture}[scale = 1] 
\tikzstyle{node_style} = [circle,draw=blue,fill=blue!30!,inner sep = 1.2pt] 
\tikzstyle{edge_style} = [gray!50, bend left = 10]
\node[node_style, label={idolssuomi} ] (v0) at (1.95428536805,1.64403252164) {};
\node[node_style, label={isänpäivä} ] (v1) at (-0.15740084419,2.44338152376) {};
\node[node_style, label={jännä} ] (v2) at (-0.99642067181,-2.5) {};
\node[node_style, label={kahvi} ] (v3) at (-1.74486132951,2.5) {};
\node[node_style, label={maikkari} ] (v4) at (1.39705780646,-0.128112819627) {};
\node[node_style, label={socialtv} ] (v5) at (-2.20615884631,0.277413438093) {};
\node[node_style, label={mtv} ] (v6) at (2.3594531145,-0.219461286257) {};
\node[node_style, label={elcrew} ] (v7) at (-0.813267801148,0.527251883445) {};
\node[node_style, label={hjallis} ] (v8) at (-1.24104110469,-1.65255253543) {};
\node[node_style, label={yle} ] (v9) at (1.49581460817,-2.25829510175) {};
\node[node_style, label={uusimtv} ] (v10) at (-0.165820093799,0.237333688992) {};
\node[node_style, label={leijonat} ] (v11) at (0.677267718833,0.00736324538757) {};
\node[node_style, label={sarasvuo} ] (v12) at (-2.5,-0.986626283097) {};
\node[node_style, label={huomenta} ] (v13) at (-1.96162253277,1.40786070973) {};
\node[node_style, label={ttk} ] (v14) at (-0.835603479575,1.89290988134) {};
\node[node_style, label={parhautta} ] (v15) at (0.719932448983,1.6863218015) {};
\node[node_style, label={alko} ] (v16) at (2.28919099853,-1.49605739487) {};
\node[node_style, label={mtv3} ] (v17) at (0.467927039812,-2.16848150594) {};
\node[node_style, label={mtvsport} ] (v18) at (0.440449115517,-1.12821294802) {};
\begin{pgfonlayer}{bg}
\draw[edge_style]  (v0) edge (v11);
\draw[edge_style]  (v0) edge (v15);
\draw[edge_style]  (v1) edge (v11);
\draw[edge_style]  (v1) edge (v14);
\draw[edge_style]  (v1) edge (v3);
\draw[edge_style]  (v5) edge (v10);
\draw[edge_style]  (v5) edge (v7);
\draw[edge_style]  (v3) edge (v13);
\draw[edge_style]  (v4) edge (v16);
\draw[edge_style]  (v4) edge (v10);
\draw[edge_style]  (v4) edge (v6);
\draw[edge_style]  (v4) edge (v7);
\draw[edge_style]  (v2) edge (v8);
\draw[edge_style]  (v2) edge (v17);
\draw[edge_style]  (v6) edge (v10);
\draw[edge_style]  (v6) edge (v16);
\draw[edge_style]  (v10) edge (v8);
\draw[edge_style]  (v10) edge (v11);
\draw[edge_style]  (v10) edge (v7);
\draw[edge_style]  (v10) edge (v13);
\draw[edge_style]  (v10) edge (v14);
\draw[edge_style]  (v10) edge (v15);
\draw[edge_style]  (v10) edge (v18);
\draw[edge_style]  (v8) edge (v12);
\draw[edge_style]  (v8) edge (v17);
\draw[edge_style]  (v7) edge (v11);
\draw[edge_style]  (v7) edge (v14);
\draw[edge_style]  (v7) edge (v12);
\draw[edge_style]  (v7) edge (v13);
\draw[edge_style]  (v7) edge (v15);
\draw[edge_style]  (v9) edge (v11);
\draw[edge_style]  (v9) edge (v17);
\draw[edge_style]  (v11) edge (v17);
\draw[edge_style]  (v11) edge (v18);
\draw[edge_style]  (v16) edge (v17);
\end{pgfonlayer}
\end{tikzpicture}

%% file: figures/Nov/filtered_2013_11_3_md3.tex
\begin{tikzpicture}[scale = 1] 
\tikzstyle{node_style} = [circle,draw=blue,fill=blue!30!,inner sep = 1.2pt] 
\tikzstyle{edge_style} = [gray!50, bend left = 10]
\node[node_style, label={kärpät} ] (v0) at (0.527995024496,0.883510998764) {};
\node[node_style] (v1) at (-0.895323775487,-2.40236083374) {};
\node[node_style, label={finland} ] (v2) at (-0.878704345693,-0.780042088093) {};
\node[node_style] (v3) at (-0.696830415456,2.33335939856) {};
\node[node_style] (v4) at (1.30412146403,-1.27922742964) {};
\node[node_style, label={ilves} ] (v5) at (0.887353588322,-2.35148715548) {};
\node[node_style] (v6) at (0.154631172065,2.09357016041) {};
\node[node_style, label={jokerit} ] (v7) at (0.771588668424,-0.0420531864655) {};
\node[node_style] (v8) at (-2.5,-0.450162034846) {};
\node[node_style, label={startups} ] (v9) at (-1.10526089045,-1.5) {};
\node[node_style, label={ylelle} ] (v10) at (-0.974326828183,2.18516936144) {};
\node[node_style, label={travel} ] (v12) at (-1.65656983803,-0.787496698669) {};
\node[node_style, label={hctps} ] (v13) at (-0.616371416663,0.753498857609) {};
\node[node_style] (v14) at (-1.43704569856,0.462683049533) {};
\node[node_style, label={tappara} ] (v15) at (1.10537821339,-1.83962226899) {};
\node[node_style, label={yle} ] (v16) at (0.205985007092,1.62242784301) {};
\node[node_style, label={cop19} ] (v17) at (1.33505418636,1.15357113699) {};
\node[node_style, label={ässät} ] (v18) at (1.39499609057,0.452672424282) {};
\node[node_style] (v19) at (2.10302414524,-1.41714556654) {};
\node[node_style, label={cheek} ] (v20) at (-1.19591535997,-0.0534734903194) {};
\node[node_style] (v21) at (-1.82876801474,-1.50962511566) {};
\node[node_style] (v22) at (0.15097157212,-2.49503705892) {};
\node[node_style, label={crowdfunding} ] (v23) at (-2.2361302714,-0.3051385448) {};
\node[node_style, label={slush} ] (v25) at (-2.1,-1.27892996593) {};
\node[node_style, label={nokia} ] (v26) at (-0.0634478832362,-0.188844018131) {};
\node[node_style] (v27) at (-2.01050237686,-1.29545083967) {};
\node[node_style, label={hifk} ] (v28) at (1.08482441886,-0.801789594588) {};
\node[node_style] (v29) at (-0.777687383276,-2.48764839555) {};
\node[node_style] (v30) at (0.293180128648,0.469832260133) {};
\node[node_style, label={slush2013} ] (v31) at (-1.12374794719,1.24516245067) {};
\node[node_style] (v32) at (2.41021525186,0.59426336514) {};
\node[node_style, label={slush13} ] (v33) at (-0.563191540816,0.1798306337736) {};
\node[node_style] (v34) at (-1.66146563182,1.83974305542) {};
\node[node_style, label={climate} ] (v35) at (1.97751152557,1.4710063316) {};
\node[node_style, label={liiga} ] (v36) at (0.214068828311,-0.614421094104) {};
\node[node_style, label={freethearctic30} ] (v37) at (1.96252923031,0.233049210215) {};
\node[node_style, label={helsinki} ] (v38) at (0.23131379174,-1.02619014043) {};
\node[node_style] (v39) at (-0.436065749082,2.00371696082) {};
\node[node_style, label={espooblues} ] (v40) at (2.31108663892,-0.674780056391) {};
\node[node_style, label={restaurantday} ] (v41) at (1.42246233721,-2.16650544447) {};
\node[node_style, label={elcrew} ] (v42) at (-2.32584704669,1.23419807643) {};
\node[node_style, label={muuten} ] (v44) at (1.20978011039,-0.262526005449) {};
\node[node_style, label={digitalist} ] (v45) at (-1.43792301194,-1.81536620323) {};
\node[node_style, label={cama} ] (v46) at (-0.100604798231,1.2) {};
\node[node_style, label={arctic30} ] (v47) at (2.5,-0.175873753399) {};
\node[node_style] (v48) at (1.24924691355,2.12464513836) {};
\node[node_style, label={cheektiedottaaettä} ] (v11) at (-0.484509269171,-1.22410812361) {};
\node[node_style, label={ravintolapäivä} ] (v24) at (1.8,-1.16706444606) {};
\node[node_style, label={movember} ] (v43) at (0.609864887939,-2.6) {};
\begin{pgfonlayer}{bg}
\draw[edge_style]  (v0) edge (v18);
\draw[edge_style]  (v0) edge (v13);
\draw[edge_style]  (v0) edge (v36);
\draw[edge_style]  (v1) edge (v11);
\draw[edge_style]  (v1) edge (v43);
\draw[edge_style]  (v2) edge (v33);
\draw[edge_style]  (v2) edge (v12);
\draw[edge_style]  (v2) edge (v26);
\draw[edge_style]  (v2) edge (v9);
\draw[edge_style]  (v2) edge (v38);
\draw[edge_style]  (v4) edge (v7);
\draw[edge_style]  (v4) edge (v36);
\draw[edge_style]  (v5) edge (v36);
\draw[edge_style]  (v5) edge (v22);
\draw[edge_style]  (v5) edge (v15);
\draw[edge_style]  (v6) edge (v33);
\draw[edge_style]  (v6) edge (v16);
\draw[edge_style]  (v7) edge (v36);
\draw[edge_style]  (v7) edge (v11);
\draw[edge_style]  (v7) edge (v40);
\draw[edge_style]  (v7) edge (v16);
\draw[edge_style]  (v7) edge (v46);
\draw[edge_style]  (v7) edge (v44);
\draw[edge_style]  (v8) edge (v12);
\draw[edge_style]  (v8) edge (v23);
\draw[edge_style]  (v9) edge (v33);
\draw[edge_style]  (v9) edge (v25);
\draw[edge_style]  (v38) edge (v33);
\draw[edge_style]  (v38) edge (v12);
\draw[edge_style]  (v38) edge (v41);
\draw[edge_style]  (v38) edge (v37);
\draw[edge_style]  (v11) edge (v20);
\draw[edge_style]  (v13) edge (v36);
\draw[edge_style]  (v13) edge (v14);
\draw[edge_style]  (v14) edge (v36);
\draw[edge_style]  (v15) edge (v28);
\draw[edge_style]  (v15) edge (v36);
\draw[edge_style]  (v43) edge (v24);
\draw[edge_style]  (v43) edge (v41);
\draw[edge_style]  (v34) edge (v42);
\draw[edge_style]  (v34) edge (v10);
\draw[edge_style]  (v17) edge (v33);
\draw[edge_style]  (v17) edge (v35);
\draw[edge_style]  (v17) edge (v37);
\draw[edge_style]  (v18) edge (v28);
\draw[edge_style]  (v18) edge (v36);
\draw[edge_style]  (v19) edge (v40);
\draw[edge_style]  (v19) edge (v36);
\draw[edge_style]  (v20) edge (v42);
\draw[edge_style]  (v20) edge (v26);
\draw[edge_style]  (v20) edge (v36);
\draw[edge_style]  (v21) edge (v33);
\draw[edge_style]  (v21) edge (v25);
\draw[edge_style]  (v22) edge (v36);
\draw[edge_style]  (v23) edge (v33);
\draw[edge_style]  (v23) edge (v25);
\draw[edge_style]  (v24) edge (v47);
\draw[edge_style]  (v24) edge (v26);
\draw[edge_style]  (v24) edge (v41);
\draw[edge_style]  (v24) edge (v37);
\draw[edge_style]  (v25) edge (v33);
\draw[edge_style]  (v25) edge (v29);
\draw[edge_style]  (v25) edge (v45);
\draw[edge_style]  (v26) edge (v33);
\draw[edge_style]  (v26) edge (v28);
\draw[edge_style]  (v26) edge (v31);
\draw[edge_style]  (v3) edge (v31);
\draw[edge_style]  (v3) edge (v10);
\draw[edge_style]  (v29) edge (v45);
\draw[edge_style]  (v30) edge (v33);
\draw[edge_style]  (v30) edge (v36);
\draw[edge_style]  (v31) edge (v33);
\draw[edge_style]  (v31) edge (v42);
\draw[edge_style]  (v32) edge (v37);
\draw[edge_style]  (v32) edge (v47);
\draw[edge_style]  (v40) edge (v44);
\draw[edge_style]  (v33) edge (v39);
\draw[edge_style]  (v33) edge (v16);
\draw[edge_style]  (v33) edge (v45);
\draw[edge_style]  (v33) edge (v44);
\draw[edge_style]  (v28) edge (v36);
\draw[edge_style]  (v35) edge (v37);
\draw[edge_style]  (v35) edge (v48);
\draw[edge_style]  (v36) edge (v27);
\draw[edge_style]  (v36) edge (v46);
\draw[edge_style]  (v36) edge (v44);
\draw[edge_style]  (v37) edge (v47);
\draw[edge_style]  (v27) edge (v45);
\draw[edge_style]  (v39) edge (v16);
\draw[edge_style]  (v10) edge (v46);
\draw[edge_style]  (v10) edge (v16);
\draw[edge_style]  (v16) edge (v48);
\end{pgfonlayer}
\end{tikzpicture}

%% file: figures/Nov/filtered_2013_11_4_md1.tex
\begin{tikzpicture}[scale = 1] 
\tikzstyle{node_style} = [circle,draw=blue,fill=blue!30!,inner sep = 1.2pt] 
\tikzstyle{edge_style} = [gray!50, bend left = 10]
\node[node_style, label={kärpät} ] (v0) at (1.65475983743,-1.07794545839) {};
\node[node_style, label={ilves} ] (v1) at (-2.5,-0.748136491847) {};
\node[node_style, label={lukko} ] (v2) at (-1.3279638732,2.5) {};
\node[node_style, label={ruutufi} ] (v3) at (1.99814228238,0.698370658937) {};
\node[node_style, label={espooblues} ] (v4) at (-2.38856759978,2.23632111567) {};
\node[node_style, label={pelicansfi} ] (v5) at (0.631144016837,-1.83600719721) {};
\node[node_style, label={hctps} ] (v6) at (-2.39202455507,1.17193467534) {};
\node[node_style, label={tappara} ] (v7) at (0.578595011587,-0.849392924578) {};
\node[node_style, label={hifk} ] (v8) at (1.38415380718,0.202350307104) {};
\node[node_style, label={jokerit} ] (v9) at (-0.290013798854,-0.501900614348) {};
\node[node_style, label={liiga} ] (v10) at (-0.358892166457,0.415999068171) {};
\node[node_style, label={kalpa} ] (v11) at (0.480299141622,0.872291894309) {};
\node[node_style, label={hymiö} ] (v12) at (-1.85450846317,0.335030856842) {};
\node[node_style, label={eteläpääty} ] (v13) at (-1.31758154514,-1.36342654711) {};
\node[node_style, label={jypliiga} ] (v14) at (-1.4442209585,-2.5) {};
\node[node_style, label={nelonenpro} ] (v15) at (1.66707187384,1.78577285313) {};
\begin{pgfonlayer}{bg}
\draw[edge_style]  (v0) edge (v8);
\draw[edge_style]  (v0) edge (v10);
\draw[edge_style]  (v0) edge (v5);
\draw[edge_style]  (v0) edge (v7);
\draw[edge_style]  (v1) edge (v6);
\draw[edge_style]  (v1) edge (v14);
\draw[edge_style]  (v1) edge (v10);
\draw[edge_style]  (v2) edge (v6);
\draw[edge_style]  (v2) edge (v10);
\draw[edge_style]  (v3) edge (v15);
\draw[edge_style]  (v3) edge (v10);
\draw[edge_style]  (v3) edge (v7);
\draw[edge_style]  (v4) edge (v6);
\draw[edge_style]  (v4) edge (v10);
\draw[edge_style]  (v5) edge (v10);
\draw[edge_style]  (v5) edge (v9);
\draw[edge_style]  (v6) edge (v12);
\draw[edge_style]  (v6) edge (v10);
\draw[edge_style]  (v7) edge (v11);
\draw[edge_style]  (v7) edge (v10);
\draw[edge_style]  (v7) edge (v14);
\draw[edge_style]  (v8) edge (v15);
\draw[edge_style]  (v8) edge (v10);
\draw[edge_style]  (v8) edge (v9);
\draw[edge_style]  (v9) edge (v10);
\draw[edge_style]  (v9) edge (v11);
\draw[edge_style]  (v9) edge (v12);
\draw[edge_style]  (v9) edge (v13);
\draw[edge_style]  (v10) edge (v11);
\draw[edge_style]  (v10) edge (v12);
\draw[edge_style]  (v10) edge (v13);
\draw[edge_style]  (v10) edge (v15);
\end{pgfonlayer}
\end{tikzpicture}

%% file: figures/Dec/filtered_2013_12_1_md1.tex
\begin{tikzpicture}[scale = 1] 
\tikzstyle{node_style} = [circle,draw=blue,fill=blue!30!,inner sep = 1.2pt] 
\tikzstyle{edge_style} = [gray!50, bend left = 10]
\node[node_style, label={koulu} ] (v0) at (2.2082551942,-1.39947966138) {};
\node[node_style, label={december} ] (v1) at (-2.07373671758,1.27437704302) {};
\node[node_style, label={itsenäisyyspäivä} ] (v2) at (0.235766471464,0.819782330314) {};
\node[node_style, label={pisa} ] (v3) at (1.73290239874,-1.89102643209) {};
\node[node_style, label={independenceday} ] (v4) at (-1.93130972677,-1.74155273765) {};
\node[node_style, label={finland} ] (v6) at (-0.823024633445,-0.804023252395) {};
\node[node_style, label={snowing} ] (v7) at (-2.35112071259,-0.778757136263) {};
\node[node_style, label={pisa2012} ] (v8) at (0.904915311753,-2.25052551999) {};
\node[node_style, label={snow} ] (v9) at (-0.567037400546,-1.2315223336) {};
\node[node_style, label={christmas} ] (v10) at (-1.77182114716,0.28413081768) {};
\node[node_style, label={happy} ] (v12) at (-1.52134469598,-1.99328932521) {};
\node[node_style, label={sixtynorth} ] (v13) at (-2.5,0.0313245420746) {};
\node[node_style, label={juti} ] (v14) at (0.549682514791,2.47498571994) {};
\node[node_style, label={ff} ] (v15) at (-0.13477334103,2.5) {};
\node[node_style, label={finnish} ] (v16) at (-0.57808343659,-2.21003724728) {};
\node[node_style, label={suomentampere} ] (v17) at (-0.185018342545,1.50876443141) {};
\node[node_style, label={helsinki} ] (v18) at (-1.02610523395,-0.105494731427) {};
\node[node_style, label={suomi} ] (v19) at (0.427493002334,-0.926147761287) {};
\node[node_style, label={ylelle} ] (v20) at (2.08622054305,1.46108795101) {};
\node[node_style, label={yle} ] (v21) at (1.21680679779,1.05) {};
\node[node_style, label={linnanjuhlat} ] (v5) at (1.10773896393,2.05440080151) {};
\node[node_style, label={independence} ] (v11) at (-0.050970621293,-2.5) {};
\begin{pgfonlayer}{bg}
\draw[edge_style]  (v0) edge (v8);
\draw[edge_style]  (v0) edge (v3);
\draw[edge_style]  (v1) edge (v10);
\draw[edge_style]  (v1) edge (v18);
\draw[edge_style]  (v2) edge (v5);
\draw[edge_style]  (v2) edge (v19);
\draw[edge_style]  (v2) edge (v6);
\draw[edge_style]  (v2) edge (v15);
\draw[edge_style]  (v2) edge (v17);
\draw[edge_style]  (v3) edge (v19);
\draw[edge_style]  (v3) edge (v8);
\draw[edge_style]  (v4) edge (v6);
\draw[edge_style]  (v4) edge (v16);
\draw[edge_style]  (v5) edge (v17);
\draw[edge_style]  (v5) edge (v20);
\draw[edge_style]  (v5) edge (v14);
\draw[edge_style]  (v5) edge (v21);
\draw[edge_style]  (v5) edge (v15);
\draw[edge_style]  (v20) edge (v21);
\draw[edge_style]  (v7) edge (v6);
\draw[edge_style]  (v7) edge (v18);
\draw[edge_style]  (v8) edge (v16);
\draw[edge_style]  (v8) edge (v6);
\draw[edge_style]  (v9) edge (v10);
\draw[edge_style]  (v9) edge (v11);
\draw[edge_style]  (v9) edge (v18);
\draw[edge_style]  (v10) edge (v6);
\draw[edge_style]  (v10) edge (v18);
\draw[edge_style]  (v11) edge (v6);
\draw[edge_style]  (v12) edge (v6);
\draw[edge_style]  (v12) edge (v18);
\draw[edge_style]  (v13) edge (v6);
\draw[edge_style]  (v13) edge (v18);
\draw[edge_style]  (v14) edge (v17);
\draw[edge_style]  (v17) edge (v6);
\draw[edge_style]  (v16) edge (v19);
\draw[edge_style]  (v16) edge (v6);
\draw[edge_style]  (v18) edge (v19);
\draw[edge_style]  (v18) edge (v6);
\draw[edge_style]  (v18) edge (v21);
\draw[edge_style]  (v19) edge (v6);
\end{pgfonlayer}
\end{tikzpicture}

%% file: figures/Dec/filtered_2013_12_2_md1.tex
\begin{tikzpicture}[scale = 1] 
\tikzstyle{node_style} = [circle,draw=blue,fill=blue!30!,inner sep = 1.2pt] 
\tikzstyle{edge_style} = [gray!50, bend left = 10]
\node[node_style, label={duudsonit} ] (v0) at (0.602779743937,2.5) {};
\node[node_style, label={linnajuhlat} ] (v1) at (1.82339778009,2.1806236829) {};
\node[node_style, label={kiakkovieraat} ] (v2) at (-0.583289799973,-2.5) {};
\node[node_style, label={kotikatsomo} ] (v3) at (1.70246064341,-0.671253423778) {};
\node[node_style, label={suomentampere} ] (v4) at (1.00948513217,-1.52205085567) {};
\node[node_style, label={kiakkovierasjuhlat} ] (v5) at (0.30873453806,-2.17850864972) {};
\node[node_style, label={helsinki} ] (v6) at (-2.2241664956,0.772617293511) {};
\node[node_style, label={sixtynorth} ] (v7) at (-2.05226674369,-1.05357702335) {};
\node[node_style, label={linnanjatkot} ] (v8) at (1.73041601094,-1.7067901616) {};
\node[node_style, label={finland} ] (v9) at (-1.72148486584,0.167904150429) {};
\node[node_style, label={itsenäisyyspäivä} ] (v10) at (1.18197176086,0.128928704472) {};
\node[node_style, label={tampere} ] (v11) at (-0.994837288338,-1.36680037679) {};
\node[node_style, label={suomi} ] (v12) at (-2.5,-0.137756926682) {};
\node[node_style, label={yle} ] (v13) at (0.00566814223072,-1.04841234173) {};
\node[node_style, label={slottsbalen} ] (v14) at (1.71585814469,1.13517332943) {};
\node[node_style, label={linnanjuhlat} ] (v15) at (0.0565154703819,-0.125747947007) {};
\node[node_style, label={fail} ] (v16) at (2.46760569504,0.823147598013) {};
\node[node_style, label={jatkot} ] (v17) at (2.34139635505,-0.327353507003) {};
\node[node_style, label={finlandia} ] (v18) at (0.5846231384,1.65999812374) {};
\node[node_style, label={christmas} ] (v19) at (-2.40784350238,1.62398105607) {};
\node[node_style, label={sibelius} ] (v20) at (-1.19144281152,1.44298474639) {};
\begin{pgfonlayer}{bg}
\draw[edge_style]  (v0) edge (v15);
\draw[edge_style]  (v0) edge (v1);
\draw[edge_style]  (v1) edge (v10);
\draw[edge_style]  (v1) edge (v14);
\draw[edge_style]  (v9) edge (v6);
\draw[edge_style]  (v9) edge (v7);
\draw[edge_style]  (v9) edge (v11);
\draw[edge_style]  (v9) edge (v12);
\draw[edge_style]  (v9) edge (v15);
\draw[edge_style]  (v9) edge (v19);
\draw[edge_style]  (v9) edge (v20);
\draw[edge_style]  (v3) edge (v15);
\draw[edge_style]  (v3) edge (v10);
\draw[edge_style]  (v4) edge (v15);
\draw[edge_style]  (v4) edge (v10);
\draw[edge_style]  (v4) edge (v5);
\draw[edge_style]  (v5) edge (v15);
\draw[edge_style]  (v5) edge (v8);
\draw[edge_style]  (v5) edge (v2);
\draw[edge_style]  (v5) edge (v11);
\draw[edge_style]  (v6) edge (v15);
\draw[edge_style]  (v6) edge (v12);
\draw[edge_style]  (v6) edge (v19);
\draw[edge_style]  (v7) edge (v15);
\draw[edge_style]  (v8) edge (v15);
\draw[edge_style]  (v8) edge (v17);
\draw[edge_style]  (v2) edge (v15);
\draw[edge_style]  (v10) edge (v15);
\draw[edge_style]  (v10) edge (v13);
\draw[edge_style]  (v11) edge (v15);
\draw[edge_style]  (v11) edge (v13);
\draw[edge_style]  (v12) edge (v15);
\draw[edge_style]  (v13) edge (v15);
\draw[edge_style]  (v20) edge (v15);
\draw[edge_style]  (v15) edge (v16);
\draw[edge_style]  (v15) edge (v17);
\draw[edge_style]  (v15) edge (v18);
\draw[edge_style]  (v15) edge (v14);
\draw[edge_style]  (v16) edge (v14);
\draw[edge_style]  (v17) edge (v14);
\draw[edge_style]  (v18) edge (v14);
\end{pgfonlayer}
\end{tikzpicture}

%% file: figures/Dec/filtered_2013_12_3_md1.tex
\begin{tikzpicture}[scale = 1] 
\tikzstyle{node_style} = [circle,draw=blue,fill=blue!30!,inner sep = 1.2pt] 
\tikzstyle{edge_style} = [gray!50, bend left = 10]
\node[node_style, label={kärpät} ] (v0) at (-2.296418252,0.638648684766) {};
\node[node_style, label={ilves} ] (v1) at (-0.583455443759,-1.40833623064) {};
\node[node_style, label={hctps} ] (v2) at (-1.87959153204,-0.952785953576) {};
\node[node_style, label={nhlfi} ] (v3) at (0.996959872543,0.670510926958) {};
\node[node_style, label={openicetre} ] (v4) at (-0.0226790515518,0.741979605624) {};
\node[node_style, label={promagnu} ] (v5) at (1.54665616385,1.66634582173) {};
\node[node_style, label={espooblues} ] (v6) at (-2.5,-0.535300492429) {};
\node[node_style, label={tappara} ] (v8) at (0.341846037347,-1.66532532047) {};
\node[node_style, label={tampere} ] (v9) at (0.322258979618,-0.449041894945) {};
\node[node_style, label={leijonat} ] (v10) at (1.32836928697,-1.07911211766) {};
\node[node_style, label={jokerit} ] (v11) at (-0.868884301788,0.55654626276) {};
\node[node_style, label={liiga} ] (v12) at (-0.545368417242,-0.52186941493) {};
\node[node_style, label={hymiö} ] (v13) at (2.10149824934,1.34991845594) {};
\node[node_style, label={muuten} ] (v14) at (0.44732868531,1.95384562522) {};
\node[node_style, label={cama} ] (v15) at (1.3891563146,0.016804187803) {};
\node[node_style, label={ässät} ] (v16) at (-1.81171019011,-1.88830355131) {};
\node[node_style, label={jypliiga} ] (v17) at (-0.354526665629,-2.5) {};
\node[node_style, label={sotshi} ] (v18) at (2.5,0.448506634042) {};
\node[node_style, label={pelicansfi} ] (v7) at (1.02812650453,-1.93564247935) {};
\begin{pgfonlayer}{bg}
\draw[edge_style]  (v0) edge (v11);
\draw[edge_style]  (v0) edge (v6);
\draw[edge_style]  (v0) edge (v12);
\draw[edge_style]  (v1) edge (v2);
\draw[edge_style]  (v1) edge (v7);
\draw[edge_style]  (v1) edge (v4);
\draw[edge_style]  (v1) edge (v8);
\draw[edge_style]  (v1) edge (v12);
\draw[edge_style]  (v6) edge (v16);
\draw[edge_style]  (v6) edge (v12);
\draw[edge_style]  (v3) edge (v12);
\draw[edge_style]  (v3) edge (v10);
\draw[edge_style]  (v3) edge (v11);
\draw[edge_style]  (v3) edge (v18);
\draw[edge_style]  (v3) edge (v15);
\draw[edge_style]  (v3) edge (v13);
\draw[edge_style]  (v3) edge (v14);
\draw[edge_style]  (v4) edge (v5);
\draw[edge_style]  (v4) edge (v12);
\draw[edge_style]  (v4) edge (v9);
\draw[edge_style]  (v11) edge (v15);
\draw[edge_style]  (v11) edge (v2);
\draw[edge_style]  (v11) edge (v12);
\draw[edge_style]  (v2) edge (v12);
\draw[edge_style]  (v2) edge (v16);
\draw[edge_style]  (v7) edge (v12);
\draw[edge_style]  (v7) edge (v15);
\draw[edge_style]  (v8) edge (v10);
\draw[edge_style]  (v8) edge (v12);
\draw[edge_style]  (v8) edge (v9);
\draw[edge_style]  (v9) edge (v12);
\draw[edge_style]  (v10) edge (v12);
\draw[edge_style]  (v10) edge (v18);
\draw[edge_style]  (v10) edge (v17);
\draw[edge_style]  (v5) edge (v15);
\draw[edge_style]  (v5) edge (v18);
\draw[edge_style]  (v12) edge (v16);
\draw[edge_style]  (v12) edge (v17);
\draw[edge_style]  (v12) edge (v14);
\draw[edge_style]  (v18) edge (v13);
\draw[edge_style]  (v17) edge (v16);
\draw[edge_style]  (v13) edge (v14);
\end{pgfonlayer}
\end{tikzpicture}

%% file: figures/Dec/filtered_2013_12_4_md1.tex
\begin{tikzpicture}[scale = 1] 
\tikzstyle{node_style} = [circle,draw=blue,fill=blue!30!,inner sep = 1.2pt] 
\tikzstyle{edge_style} = [gray!50, bend left = 10]
\node[node_style, label={änäri} ] (v0) at (1.28580047248,-1.25475283355) {};
\node[node_style, label={liiga} ] (v1) at (-1.57257149836,-1.53707623311) {};
\node[node_style, label={nhlfi} ] (v2) at (0.219411940679,-1.06053589114) {};
\node[node_style, label={helsinki} ] (v4) at (-0.566940687183,2.5) {};
\node[node_style, label={suomi} ] (v5) at (1.11552280055,-0.175786095256) {};
\node[node_style, label={travel} ] (v6) at (-2.06780555973,2.08951297536) {};
\node[node_style, label={finland} ] (v7) at (0.879481981853,1.57045711072) {};
\node[node_style, label={tappara} ] (v8) at (-1.98032451453,-2.10561392118) {};
\node[node_style, label={leijonat} ] (v9) at (-0.727831436378,-2.5) {};
\node[node_style, label={jokerit} ] (v10) at (-2.5,-1.60882013413) {};
\node[node_style, label={u20fi} ] (v11) at (-0.677775838756,-1.71167713348) {};
\node[node_style, label={u20} ] (v12) at (-0.183992884348,-2.66696377023) {};
\node[node_style, label={muuten} ] (v13) at (-1.84889919624,-0.641109746074) {};
\node[node_style, label={newyear} ] (v14) at (0.665333208178,2.0) {};
\node[node_style, label={joulu} ] (v15) at (-1.27405880826,0.630576447631) {};
\node[node_style, label={hymiö} ] (v16) at (-0.768525354914,-0.680698794953) {};
\node[node_style, label={christmas} ] (v17) at (-0.905323315437,1.83672743514) {};
\node[node_style, label={sotshi} ] (v18) at (0.793059503968,-2.01550007229) {};
\node[node_style, label={nye} ] (v19) at (-1.40851255978,2.41231391174) {};
\node[node_style, label={uusivuosi} ] (v3) at (-0.1415969864,2.28890432125) {};
\begin{pgfonlayer}{bg}
\draw[edge_style]  (v9) edge (v1);
\draw[edge_style]  (v9) edge (v12);
\draw[edge_style]  (v9) edge (v11);
\draw[edge_style]  (v9) edge (v18);
\draw[edge_style]  (v1) edge (v11);
\draw[edge_style]  (v1) edge (v15);
\draw[edge_style]  (v1) edge (v8);
\draw[edge_style]  (v1) edge (v10);
\draw[edge_style]  (v1) edge (v12);
\draw[edge_style]  (v1) edge (v13);
\draw[edge_style]  (v2) edge (v11);
\draw[edge_style]  (v2) edge (v5);
\draw[edge_style]  (v2) edge (v0);
\draw[edge_style]  (v2) edge (v13);
\draw[edge_style]  (v2) edge (v16);
\draw[edge_style]  (v2) edge (v18);
\draw[edge_style]  (v3) edge (v15);
\draw[edge_style]  (v3) edge (v4);
\draw[edge_style]  (v3) edge (v14);
\draw[edge_style]  (v3) edge (v19);
\draw[edge_style]  (v4) edge (v6);
\draw[edge_style]  (v4) edge (v7);
\draw[edge_style]  (v4) edge (v14);
\draw[edge_style]  (v4) edge (v17);
\draw[edge_style]  (v4) edge (v19);
\draw[edge_style]  (v5) edge (v7);
\draw[edge_style]  (v5) edge (v11);
\draw[edge_style]  (v6) edge (v17);
\draw[edge_style]  (v6) edge (v19);
\draw[edge_style]  (v7) edge (v17);
\draw[edge_style]  (v7) edge (v14);
\draw[edge_style]  (v8) edge (v10);
\draw[edge_style]  (v8) edge (v11);
\draw[edge_style]  (v17) edge (v15);
\draw[edge_style]  (v0) edge (v16);
\draw[edge_style]  (v0) edge (v18);
\draw[edge_style]  (v10) edge (v16);
\draw[edge_style]  (v10) edge (v11);
\draw[edge_style]  (v12) edge (v11);
\draw[edge_style]  (v16) edge (v15);
\draw[edge_style]  (v16) edge (v13);
\end{pgfonlayer}
\end{tikzpicture}

%% file: related.tex
\section{Related work}
\label{section:related}

Partitioning a graph in dense subgraphs is a well-established problem.
Many of the existing works adopt as density definition the average-degree notion~\cite{andersen2009finding, khuller2009finding, galbrun2014overlapping, tsourakakis2014novel}.
The densest subgraph, under this definition, can be found in polynomial time~\cite{goldberg1984finding}. 
Moreover, there is a 2-approximation greedy algorithm by Charikar~\cite{charikar2000greedy} and Asahiro~\cite{asahiro2000greedily}, which runs in linear time of the graph size. 
Many recent works develop methods to maintain the average-degree densest-subgraph in a streaming scenario~\cite{bhattacharya2015space, epasto2015efficient, esfandiari2015applications, mcgregor2015densest, mitzenmacher2015scalable}. Alternative density definitions, such as variants of quasi-clique, are often hard to approximate or solve by efficient heuristics due to connections to \np-complete Maximum Clique problem~\cite{makino2004new, alvarez2006large, tsourakakis2013denser}.


A line of work focuses on dynamic graphs, 
which model node/edge additions/deletions. 
Different aspects of network evolution, 
including evolution of dense groups, were studied in this setting 
\cite{backstrom2006group, berlingerio2009mining, myers2014bursty, li2014efficient}. 
However, here we use the interaction-network model, 
which is different to dynamic graphs, as it captures the instantaneous interactions between nodes.

Another classic approach to model temporal graphs is to consider graph snapshots, 
find structures in each snapshot separately (or by incorporating information from previous snapshots), 
and then summarize historical behavior of the discovered structures~\cite{lin2008facetnet, asur2009event,greene2010tracking, mucha2010community, berlingerio2013abacus}.
These approaches usually focus on the temporal coherence of the dense structures 
discovered in the snapshots and assume that the snapshots are given. 
In this work we aggregate instantaneous interaction into timeline partitions of arbitrary lengths.

To the best our knowledge, the following works are better aligned with our approach.
A work of Rozenshtein et al.~\cite{rozenshtein2017finding} 
considers a problem of finding the densest subgraph in a temporal network. 
However, first, they do not aim on creating a temporal partitioning.
Second, they are interested in finding a single dense subgraph whose edges occur in $k$ short time intervals.
On the contrary, in this work we search for an interval partitioning and 
consider only graphs that are span continuous intervals. 
Other close works are by Jethava and Beerenwinkel~\cite{jethava2015finding} and Semertzidis et al.~\cite{semertzidis2016best}. 
However, these works consider a set of snapshots and search for a single heavy subgraph induced by one or several intervals. The work of Semertzidis et al.~\cite{semertzidis2016best} explores different formulations for the \emph{persistent heavy subgraph problem}, including maximum average density, while Jethava and Beerenwinkel \cite{jethava2015finding} focus solely on maximum average density.


%% file: concl.tex
\section{Conclusions}
\label{section:conclusions}

In this work we consider the problem of 
finding a sequence of dense subgraphs in a temporal network. 
We search for a partition of the network timeline into $k$ non-overlapping intervals, 
such that the intervals span subgraphs with maximum total density. 
To provide a fast solution for this problem we adapt recent work on dynamic densest subgraph and approximate dynamic programming. In order to ensure that the episodes we discover consist of a diverse set of nodes, we adjust the problem formulation to encourage coverage of a larger set of nodes. While the modified problem is \np-hard, we provide a greedy heuristic, which performs well on empirical tests.

The problems of temporal event detection and timeline segmentation can be formulated in various ways depending on the type of structures that are considered to be interesting. Here we propose segmentation with respect to maximizing subgraph density. The intuition is that those dense subgraphs provide a sequence of interesting events that occur in the lifetime of the temporal network. However, other notions of interesting structures, such as frequency of the subgraphs, or statistical non-randomness of the subgraphs, can be considered for future work. In addition, it could be meaningful to allow more than one structure per interval. Another possible extension is to consider overlapping intervals instead of a segmentation.

%% file: appendix.tex
\appendix
\subsection{Proofs}
\label{appendix}
	
\begin{proof}[Proof of Proposition~\ref{prop:NPP2}]
	
	In the proof we show that Problem~\ref{maxgendegree} is at least as hard as \np-complete \atleastk problem: given static graph $G'=(V', E')$ and parameter $k'$, find the densest subgraph with at least $k'$ nodes. 
	
	We will consider an instance of Problem~\ref{p2} for a temporal graph with only one timestamp (static graph), $k=1$ and $\dvrs$ being a standard cover. We will refer to this instance as $P$. We abuse the notation, and write $\sum_{G_i\in \graphs}\dens(G_i)$ as $d$, $\dvrs(\graphs)$ as $c$, $\profit$ as $o$. For a given $\lambda$ denote optimal value $\profit(\lambda)$ as $o(\lambda)$, corresponding values of $c$ and $d$ as $c(\lambda)$ and $d(\lambda)$.
	
	
	Observe that $c(\lambda)$ is a non-decreasing function of $\lambda$: consider $\lambda_1 < \lambda_2$, let $S_1$ be the optimal solution for $\lambda_1$, write $o_1$ for $o(\lambda_1)$, $d_1$ for $d(\lambda_1)$ and $c_1$ for $c(\lambda_1)$. Profit value of $S_1$ is $o_1=d_1 + \lambda_1c_1$. Similarly, for $\lambda_2$ define optimal solution $S_2$ with profit $o_2=d_2 +\lambda_2c_2$. Suppose that $c_2<c_1$. Then, the only option is $d_2 > d_1$, otherwise $S_1$ would provide a better solution for $\lambda_2$, while $S_2$ is optimum. For this remaining case $d_2 > d_1$, from optimality of $S_1$ and $S_2$ we have: $d_1 + \lambda_1 c_1\geq d_2 + \lambda_1 c_2$ and $d_2 + \lambda_2 c_2\geq d_1 + \lambda_2 c_1$. Thus, $\lambda_2\leq (d_2-d_1)/(c_1-c_2) \leq \lambda_1$, which leads to the contradiction with $\lambda_1<\lambda_2$. This concludes monotonicity of $c(\lambda)$.
	
	Next, note that due to optimality, any optimum solution $S$ with density $d$ and cover $c$ has $d$ equal to the maximum density of a graph with at least $c$ nodes. Furthermore, for every fixed natural $k'\leq n$ (where $n$ is the number of nodes) there exists $\lambda$, such that $c(\lambda) \geq k'$. For $\lambda=n$ optimum solution $S$ is guaranteed to have $c=n$: consider solution $S_1$ with $c_1=n$ and $d_1=\dens(G)$ (density of the whole graph) obtained for some fixed $\lambda$. Due to optimality of $S_1$ it holds that $\dens(G)+\lambda n\geq d(S_2)+\lambda c(S_2)$ for any other subgraph $S_2$. Thus, 
	\[\lambda\geq\frac{d(S_2)- d(G)}{n-c(S_2)}.\] 
	Since $c(\lambda)$ is monotone, any larger $\lambda$, e.g., $\lambda^* = n \geq \frac{d(S_2)- d(G)}{n-c(S_2)}$, will guarantee $c(\lambda) = n$.
	
	Now, given an instance of \atleastk with a static graph $G'=(V',E')$ and $k'$, we can solve it doing a binary search for $\lambda$ and thus solving a polynomial ($\log n$) number of instances of $P$ were the only timestamp $t_1$ contains all edges from $G'$.
	
\end{proof}	

\begin{proof}[Proof of Proposition~\ref{prop:NPP5}]
	Similar to the proof for Problem~\ref{p2}.
\end{proof}

\begin{proof}[Proof of Proposition~\ref{prop:statgreedy}]
	Let $H^* = (W, A)$ be the optimal subgraph. 
	
	First, for each $v\in W$ holds $\g(v\mid H^*)\geq\gdens(H^*)$. To see this, note that 
	\[
	\gdens(H^*)=\frac{\sum_{u\in W} \g(u\mid H^*)}{2|W|}.
	\]
	Since $H^*$ is optimal and $\g$ is increasing w.r.t node addition, 
	\[
	\frac{\sum_{u\in W \setminus\{v\}} \g(u\mid H^*)}{2(|W|-1)}\leq \frac{\sum_{u\in V(H^*)} \g(u\mid H^*)}{2|W|}.
	\]
	Solving $\g(v\mid H^*)$ leads to
	\[
	\g(v\mid H^*)\geq \frac{\sum_{u\in V(H^*)} \g(u\mid H^*)}{|V(H^*)|}\geq \gdens(H^*).
	\]
	
	The rest of the argument follows the classic proof by
	\cite{khuller2009finding}. Denote $\gdens(H^*)$ as $O$. Consider iteration
	$i$, when the first node $v\in W$ is removed. Let $\bar{H}=
	(\bar{V},\bar{E})$ be the remaining graph after iteration $i$. By greedy
	construction all vertices $v\in\bar{V}$ have weight $\g(v \mid \bar{H}) \geq \g(v\mid H^*)\geq O$
	and $\gdens(\bar H)\geq O|\bar{V}|/2|\bar{V}| = O/2$. As greedy outputs
	the best subgraph, it will always output a subgraph with weight no worse
	than $H^*$.
\end{proof}	

\subsection{Incremental $k$-densest subgraphs with generalized average degree}
\label{app:incremental}

Given a stream of incremental edge updates to graph $H$ we would like to find and keep up-to-date a subgraph $H_i$, which maximizes $\gdens(H_i)$ for some generalized degree function $\g(u,v\mid H_i)$. 

To keep $H_i$ updated we can use the data structure and update procedure designed for the densest subgraph by Epasto et al.~\cite{epasto2015efficient}. Here we briefly describe it for the sake of completeness and discuss necessary modifications.

The approach uses the following variant of the greedy algorithm $\find$ as a building block. Additionally to the graph $H$, $\find$ requires parameters $\beta$ and $\epsilon$ as an input. Parameter $\beta$ has a meaning of the estimate for the optimal profit and $\epsilon$ is accuracy. 

\begin{algorithm}[]
	\KwData{static graph $H=(V,\estat)$, $\beta>0$, $\epsilon>0$}
	$H_0, \bar H=V(H)$;  $t=0$\;
	\While{$H_t\not = \emptyset$ and $t\leq \ceil{\log_{1+\epsilon}(|V(H_t)|)}	$}
	{
		$A(H_t) = \{v\in V(H_t): \g(v\mid H_t) < 2(1+\epsilon)\beta\}$\;
		$H_{t+1}=H_t\setminus A(H_t)$\;		
		\lIf {$\gdens(H_{t+1}) > \gdens(\bar H)$} {$\bar H = H_{t+1}$}
		$t=t+1$
	}	
	\Return $\bar H$
	\caption{\texttt {Find}}
	\label{alg:find}
\end{algorithm}

Algorithm~\ref{alg:find} has a property formulated in preposition~\ref{findproperty}, which is used in the binary search for the approximate optimal subgraph in Algorithm~\ref{alg:findDensest}.

\begin{proposition}\label{findproperty}
	If $0< \beta \leq \frac{O}{2(1+\epsilon)}$ with $O$ being optimal solution, then Algorithm~\ref{alg:find} finds a subgraph with weight at least $\beta$, while if $\beta > O$ a subgraph with the profit strictly less than $\beta$ is found.
\end{proposition}

\begin{algorithm}[]
	\KwData{static graph $H=(V,\estat)$, lower bound for optimal profit $\rho>0$, $\epsilon>0$}
	\KwResult{subgraph $\bar H$ which maximizes $\gdens(\bar H)$ within factor $2(1+\epsilon)^2$}
	$\bar H=\emptyset$;
	$\beta=\max(\frac{1}{4(1+\epsilon)}, (1+\epsilon) \rho)$\;
	\While{True}
	{
		$H'=Find(H,\beta,\epsilon)$\;
		\eIf {$\gdens(H')\geq \beta$}
		{
			$\bar H = H'$; $\beta=(1+\epsilon)\gdens(\bar H)$\;	
		}
		{return $(\beta, \bar H)$}
	}
	\caption{\texttt {FindDensest}}
	\label{alg:findDensest}
\end{algorithm}

Consider the last call of $Find(H,\beta,\epsilon)$ in the Algorithm~\ref{alg:findDensest}. Let $S_t$ be the set of nodes of graph $H_t$ at the iteration $t$. The nested set of node sets $S=(S_1, \dots, S_k)$ with $k=\ceil{\log_{1+\epsilon}(|V|)}$ has the following property by construction: $S_0=V$, $S_k=\emptyset$ and for $t\in [1,k-1]$ set $S_{t+1}$ is obtained from $S_t$ by removing all nodes $v$ with $\g(v\mid H_t)<2(1+\epsilon\beta)$. Furthermore, it can be shown that if a set $S=(S_1, \dots, S_k)$ with $k=\ceil{\log_{1+\epsilon}(|V|)}$ has that property, than there is a set $S_i\in S$, such that $S$ induces a subgraphs $H$ with the profit within factor $2(1+\epsilon)^2$ of optimal. 

The update procedure Algorithm~\ref{alg:Add} is designed to keep $S$ updated. Note that new edges can only increase generalized degree of nodes, thus nodes may need to be assigned to the set with larger $t$. The changes are propagated among the neighbors. The only difference with the original procedure from~\cite{epasto2015efficient} is that we need extra care with new nodes (line 12): adding a new node $u$ affects generalized degree of all nodes and we have to push all nodes on the stack to check. If some node must be moved to the set $k = \ceil{\log_{1+\epsilon}(|V|)}$, then this will violate requirements for $S$ and $S$ is rebuild from scratch. 

%

\begin{algorithm}[]
	\KwData{graph $H=(V,\estat)$, $\beta$, $\epsilon>0$, $S=(S_0,\dots, S_k)$}
	\KwResult{Updated $S$ or indicator that $S$ must be rebuild.}
	$\estat = \estat \cup \{(u,v)\}$\;
	Update degrees of $u$ and $v$\;
	$Stack = \emptyset$\;
	$Push(u,v,Stack)$\;
	\While{$Stack\not= \emptyset$}
	{
		$s=pop(Stack)$\;
		$S_t=\{S_t \in S : s\in S_t\setminus S_{i+1}\}$\;
		\lIf{$\g(s\mid H_t)<2(1+\epsilon)\beta$}{continue}
		$t'=\min \{t' : t'>t~and~\g(s\mid H_{t'})<2(1+\epsilon)\beta\}$
		
		\lIf{$t'=\ceil{\log_{1+\epsilon(|V|)}}$}{return True}
		Add $s$ to the sets $S_{t+1}, \dots, S_{t'}$\;
		\leIf{$s\in V$}
		{Push all neighbors of $s$ to the $Stack$}
		{Push all nodes in $V$ to the $Stack$}
	}
	{return False}
	\caption{\texttt {Add}}
	\label{alg:Add}
\end{algorithm}

\begin{algorithm}[]\label{alg:Update}
	\KwData{graph $H=(V,\estat)$, $\epsilon>0$}
	\KwResult{Updated optimal subgraph $\bar H$}
	
	$(\beta,\bar H)=FindDensest(H,0,\epsilon)$ and let $(S_0,\dots,S_k)$ be the sets computed by $Find$\;
	Output $\bar H$\;
	\While{$True$}
	{
		Wait for a new edge $(u,v)$\;
		$Rebuild=Add((u,v), (S_0, \dots, S_k),H,\beta,\epsilon)$\;
		\If{$Rebuild$}
		{$(\beta,\bar H)=FindFensest(H,\beta,\epsilon)$ (update $(S_0,\dots,S_k)$)}
		Output $\bar H$\;
	}
	\caption{\texttt {Update}}
\end{algorithm}

The total number of operations, needed to keep the approximate optimal subgraph updated is the following.
First, FindDensest is done in $O(|V|^2\log D \epsilon^{-1})$, where $D$ is maximum value of average generalized degree and is $O(|V|)$. The total number of FindDensest calls for a graph $G=(V,E)$ is $O(\epsilon^{-1}\log(D))$. The total number of operations between two consecutive calls of FindDensest $O(|V|^2\epsilon^{-1}\log {|V|})$. Thus, keeping the solution updated requires $O(|V|^2\epsilon^{-2}\log^2{D}) + O(|V|^2\epsilon^{-1}\log {|V|})= O(|V|^2\epsilon^{-2}\log^2{D})$ of running time. This translates into $O(\frac{|V|^2}{|E|}\epsilon^{-2}\log^2{D})$ amortized cost per edge insertion.
Space requirements are $O(|V|+|E|)$, as for the original algorithm in~\cite{epasto2015efficient}.